\documentclass[letter, 12pt, runningheads]{article}

\usepackage{setspace}
\usepackage[ruled,vlined]{algorithm2e}
\usepackage{algorithmic}
\usepackage{amsfonts}
\usepackage{amsmath}
\usepackage{amssymb}
\usepackage{amsthm}
\usepackage[titletoc]{appendix}
\usepackage[english]{babel}
\usepackage{booktabs}
\usepackage{cancel}
\usepackage{caption}
\usepackage{comment}
\usepackage{diagbox}
\usepackage{dsfont}
\usepackage{enumerate}
\usepackage{lmodern,textcomp}
\usepackage{float}
\usepackage[left = 1in, right = 1in]{geometry}
\usepackage{graphicx}
\usepackage[colorlinks=true,linkcolor=blue,citecolor=blue]{hyperref}
\usepackage{cleveref}
\usepackage{fancyhdr}
\usepackage{indentfirst}
\usepackage[utf8]{inputenc}
\usepackage{csquotes} 
\usepackage{amssymb}
\usepackage{mathtext}
\usepackage{mathtools}
\usepackage{multicol}
\usepackage{multirow}
\usepackage[numbers]{natbib}
\usepackage{nicefrac}
\usepackage[short]{optidef}
\usepackage{pgfplots}
\usepackage{setspace}
\usepackage{subcaption}
\usepackage{supertabular}
\usepackage[flushleft]{threeparttable}
\usepackage{tikz,tkz-graph,tkz-berge}
\usetikzlibrary{decorations.pathreplacing}
\usetikzlibrary{arrows.meta}
\usepackage[textwidth=30mm]{todonotes}
\usepackage{xargs}
\usepackage{xcolor}
\usepackage{wrapfig}
\usepackage[all]{xy}

\newcommand{\SW}{\mathrm{SW}}

\newcommand{\opt}{\mathrm{OPT}}

\newcommand{\ia}{A}
\newcommand{\ib}{B}

\newcommand{\ja}{\alpha}
\newcommand{\jb}{\beta}

\newcommand{\probmu}{\hat{\mu}}

\newcommand{\esp}[1]{\mathbb E\left[#1\right]}
\newcommand{\SI}{\mathcal{I}} 
\newcommand{\matching}[2]{\mathcal{M}\left(#1,\,#2\right)}

\DeclareMathOperator{\allocation}{\mu,\boldsymbol{p}}
\DeclareMathOperator{\bi}{>}
\DeclareMathOperator*{\argmin}{argmin}
\DeclareMathOperator*{\argmax}{argmax}
\DeclareMathOperator{\ALG}{\mathrm{ALG}}

\DeclareMathOperator{\normSI}{\mathcal{J}}

\newtheorem{theorem}{Theorem}[section]
\newtheorem{proposition}[theorem]{Proposition}

\newtheorem{lemma}[theorem]{Lemma}

\newtheorem{definition}[theorem]{Definition}
\newtheorem{remark}[theorem]{Remark}
\newtheorem{example}[theorem]{Example}

\usepackage{circuitikz}
\usetikzlibrary{automata, positioning, arrows}
\usetikzlibrary{shapes,shapes.geometric}
\usepackage{tikz} 
\usepackage{soul}

\allowdisplaybreaks
\begin{document}
	\thispagestyle{empty}

	\begin{center}
		\LARGE{\textbf{Stability in Online Assignment Games}}
	\end{center}
	\medskip
	
	\begin{center}
		\textbf{Emile Martinez}\quad \textbf{Felipe Garrido-Lucero}\quad \textbf{Umberto Grandi}
	\end{center}
	
	\begin{center}
		IRIT, Université Toulouse Capitole
		
		Toulouse, France
	\end{center}
	\vspace{1cm}
	
	\begin{abstract}
		\noindent The assignment game models a housing market where buyers and sellers are matched, and transaction prices are set so that the resulting allocation is stable. Shapley and Shubik showed that every stable allocation is necessarily built on a maximum social welfare matching. In practice, however, stable allocations are rarely attainable, as matchings are often sub-optimal, particularly in online settings where eagents arrive sequentially to the market. In this paper, we introduce and compare two complementary measures of instability for allocations with sub-optimal matchings, establish their connections to the optimality ratio of the underlying matching, and use this framework to study the stability performances of randomized algorithms in online assignment games.
		\medskip
		
	\end{abstract}
	
	\section{Introduction}
	The assignment game of \citet{shapley1971assignment} is a classical model of matching with transfers in which house buyers and house sellers pair together and set transaction prices. The objective is to find a stable allocation, i.e., a matching together with a price vector such that no pair of agents has an incentive to abandon their partners and trade with each other instead. Since its introduction, the model has been extensively studied and generalized \cite{demange_strategy_1985,demange_multi_1986,eriksson_stable_2001,rochford_symmetrically_1984,crawford_job_1981,kelso_jr_job_1982,garrido2025stable,echenique_core_2002,shioura_partnership_2017}.
	
	Shapley and Shubik’s seminal work showed that stable allocations must be based on optimal matchings, i.e., matchings that maximize social welfare. However, in real-life scenarios optimal solutions are rarely attainable, for example in online markets \cite{coles1998marketplaces} and online advertising \cite{mehta2013online}, typically modeled as online matching markets where algorithms like \textit{Greedy} \cite{fisher2009analysis,lehmann2001combinatorial} and \textit{Ranking} \cite{aggarwal2011online,karp1990optimal} are commonly applied. As a result, the implemented solutions are often \textit{unstable}.
	Yet, most existing stability metrics are binary: they simply determine whether a matching is stable or not. This motivates the need for non-binary measures of stability that can assess how close a sub-optimal matching is to being stable, and their posterior application to analyze the stability of online algorithms.
	\medskip
	
	\noindent\textbf{Contributions}. The contributions of the article are as follows.
	\smallskip
	
	\noindent$\bullet$ We leverage three metrics to evaluate allocations arising from sub-optimal matchings: the optimality ratio; the stability index, derived from the subset instability of \citet{jagadeesan2021learning}; and the $\kappa$-approximate core \cite{faigle1998approximate,qiu2016approximate,vazirani2022general} from cooperative game theory.
	\smallskip
	
	\noindent$\bullet$ We prove that, for any allocation, the optimality ratio upper bounds the stability index, which in turn upper bounds the parameter $\kappa$ of the $\kappa$-approximate core. This result refines and strengthens the classical connection between stability and optimality established by Shapley and Shubik.
	\smallskip
	
	\noindent$\bullet$ We show that whenever prices can be set \textit{a posteriori} (either outside online applications or in settings where matching and pricing can be decoupled), the stability index can be maximized to exactly match the optimality ratio of the underlying matching.
	\smallskip
	
	\noindent$\bullet$ We initiate the study of randomized algorithms in \textit{online assignment games}, where either buyers or edges between buyers and sellers arrive sequentially. We evaluate stability at three levels: ex-post (worst-case realization), ex-ante (expected performance), and average (performance under expected utilities). While ex-ante and average guarantees coincide when analyzing social welfare (e.g., the competitive ratio of online algorithms), they diverge under our non-linear stability notions.
	\smallskip
	
	\noindent$\bullet$ Building on existing literature and new results, we establish lower bounds for the stability metrics at the three levels above for both vertex-arrival and edge-arrival models, and complement them with tight examples in several cases.
	\medskip
	
	\noindent\textbf{Related work}. Our article builds on two main literatures: stable matching and online matching. In the former, given agents with (possibly endogenous) preferences, the goal is to form a matching (possibly with additional elements such as prices) such that no pair of agents would prefer each other over their assigned partners. In the latter, the focus is on designing algorithms that make decisions as the market evolves and that guarantee high social welfare, ideally independent of the specific instance.
	
	Despite the extensive literature on stable matchings, most work on uncertainty has focused on settings without transfers. In economics, several authors have studied dynamically stable matchings \cite{bullinger2025stability,liu2023stability,doval2022dynamically} as well as markets with incomplete information about agents’ preferences \cite{bikhchandani2017stability,liu2014stable}. In machine learning, research often considers markets with unknown preferences and addresses them through regret minimization \cite{das2005two,liu2020competing,basu2021beyond,cen2022regret}. More recently, \citet{min2022learn} and \citet{jagadeesan2021learning} studied the assignment game with unknown utility functions, introducing instability metrics as notions of regret and designing reinforcement learning algorithms that achieve sublinear regret bounds.
	
	Online bipartite matching \cite{mehta2013online,echenique_online_2023,huang2024online} is one of the most fundamental problems in the online algorithms literature. It dates back to the seminal work of \citet{karp1990optimal}, who introduced the \textit{Ranking} algorithm and proved its optimality in the unweighted case. Since then, research on online matching has focused on designing increasingly competitive algorithms, that is, algorithms that achieve social welfare closer to that of the offline optimum that knows all arrivals in advance. \citet{aggarwal2011online} extended this result by showing that \textit{Ranking} is also optimal in the vertex-weighted case.
	
	For the more general edge-weighted matching problem, \citet{feldman2009online} introduced the free disposal assumption, noting that without it no randomized algorithm can achieve a constant competitive ratio. Under this assumption, \textit{Greedy} attains a $\nicefrac{1}{2}$-competitive guarantee. \citet{fahrbach2022edge} were the first to surpass this long-standing $\nicefrac{1}{2}$-barrier, a result subsequently improved by several works \cite{blanc2022multiway,shin_et_al,gao2022improved}.
	\medskip
	
	\noindent\textbf{Outline.} The rest of the article is organized as follows. \Cref{sec:the_assignment_game} introduces the Shapley–Shubik assignment game, presenting its main techniques and results. \Cref{sec:optimality_and_stability} presents the stability index and the approximate core, relates them to the optimality ratio, and presents a simple pricing procedure that guarantees a $\nicefrac{1}{2}$ ratio for stability. \Cref{sec:online_stable_matching} initiates the study of randomized algorithms for online assignment games in both edge-arrival and vertex-arrival models. This section defines our three levels of performance guarantees (ex-post, ex-ante, and average) and establishes lower bounds for each metric, with several tight examples. \Cref{sec:conclusions} concludes. The missing proofs are included in Appendix \ref{sec:missing proofs}.
	

	\section{The Assignment Game}\label{sec:the_assignment_game}
	
	The assignment game, in the classical notation of \citet{shapley1971assignment}, consists of a tuple $\Gamma := (B,S,\boldsymbol{h},\boldsymbol{c})$ where $B$ and $S$ are finite agents sets
	which we name buyers and sellers, $\boldsymbol{h} := (h_{i,j})_{i\in B, j \in S}$, where $h_{i,j} \bi 0$ represents the valuation of buyer $i$ for seller $j$'s house, and $\boldsymbol{c} := (c_j)_{j\in S}$, where $c_j \bi 0$ is the valuation of seller $j$ for her house. \Cref{fig:housing_market_example} shows an assignment game example with three buyers and two sellers.
	\begin{figure}[h]
		\centering
		\includegraphics[scale = 0.25]{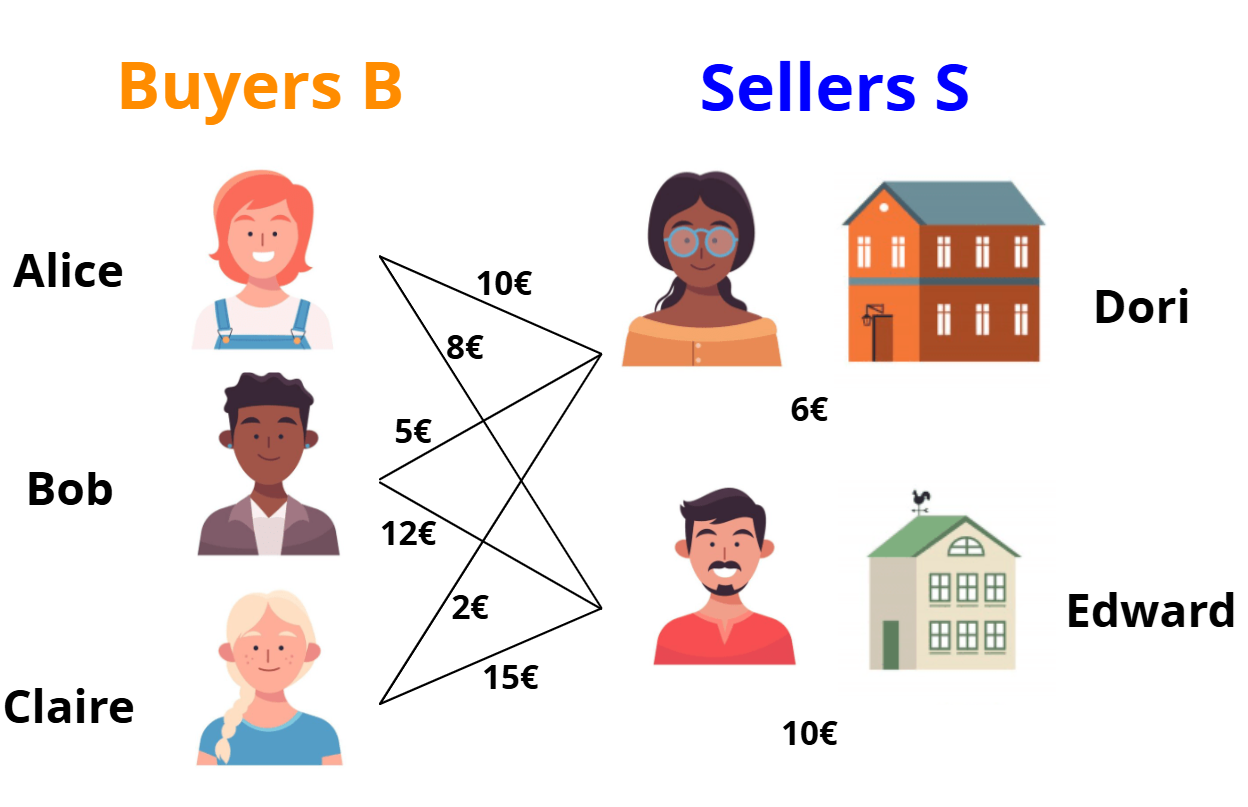}
		\caption{An assignment game instance. Buyers' valuations are denoted over the edges while sellers' valuations are denoted under their houses.}
		\label{fig:housing_market_example}
	\end{figure}
	Throughout the article, we will use $i$ to denote a typical buyer and $j$ to denote a typical a seller.
	
	
	\begin{definition}
		A \textbf{matching} is an injective function $\mu : B \cup S \to B \cup S$ such that, (1) $\mu\circ \mu = \mathrm{Id}$, (2) for any $i \in B$, $\mu(i) \in S \cup \{i\}$, and (3) for any $j \in S$, $\mu(j) \in B \cup \{j\}$. Whenever $\mu(i) = j$, for  $i \in B$ and $j\in S$, we say that $i$ and $j$ are matched, while for any agent $k \in B\cup S$, such that $\mu(k) = k$, we say that the agent is unmatched.
	\end{definition}
	
	Given a matched pair of agents $(i,j) \in B \times S$, we will alternatively write $i = \mu_j$, $j = \mu_i$, or $(i,j)\in \mu$. Finally, we denote $\mathcal{M}(B,S)$ to the set of all matchings between $B$ and $S$.
	
	\begin{definition}
		An \textbf{allocation} is a pair $(\allocation)$, where $\mu\in \mathcal{M}(B,S)$ is a matching and $\boldsymbol{p}\in\mathbb{R}_+^{S}$ is a price vector.
	\end{definition}
	
	Given an allocation $(\allocation)$ and $j$ a matched seller, $p_j \in \boldsymbol{p}$ represents the price that the buyer $\mu_j$ payed for $j$'s house. By convention, we assume $p_j = c_j$ whenever $j$ is unmatched. 
	
	\begin{definition}
		Given an allocation $(\allocation)$, $i\in B$, and $j \in S$, we define the \textbf{agents' utilities} as
		\begin{align*}
			u_i(\allocation) &:= \left\{\begin{array}{cc}
				h_{i,\mu_i} - p_{\mu_i} & \text{ if } \mu_i \neq i, \\ 
				0 & \text{ if } \mu_i = i. \\ 
			\end{array}\right.\\
			v_j(\allocation) &:= p_j - c_j.
		\end{align*}
	\end{definition}
	
	For example, suppose that Alice and Dori in \Cref{fig:housing_market_example} are matched and Alice pays 7€. Alice's utility is 3€ while Dori's utility is 1€.
	
	\begin{definition}\label{def:stability}
		An allocation $(\allocation)$ is called \textbf{stable} if it verifies
		\begin{itemize}
			\item $u_i(\allocation)\geq 0$ and $v_j(\allocation) \geq 0$ for any $i\in B$ and $j \in S$.
			\item There is no $p \in \mathbb{R}_+$ and $(i,j)\in B\times S$ such that $h_{i,j} - p \bi u_i(\allocation)$ and $p- c_j \bi v_j(\allocation)$.
		\end{itemize}
	\end{definition}
	
	
	The first condition of Definition \ref{def:stability} corresponds to individual rationality, while the second one to the non-existence of blocking pairs. In our assignment game example, as illustrated in \Cref{fig:individual_rational}, matching Bob and Edward at price $9$€ is not individually rational, as Edward prefers to be unmatched. Similarly, as illustrated in \Cref{fig:blocking_pair}, matching Bob and Edward at price $11$€ and letting Claire unmatched creates a blocking pair, Claire and Edward, as Claire can offer $12$€ and strictly increase her and Edward's utility.
	
	\begin{figure}[h]
		\begin{subfigure}{0.49\textwidth}
			\centering
			\includegraphics[scale = 0.33]{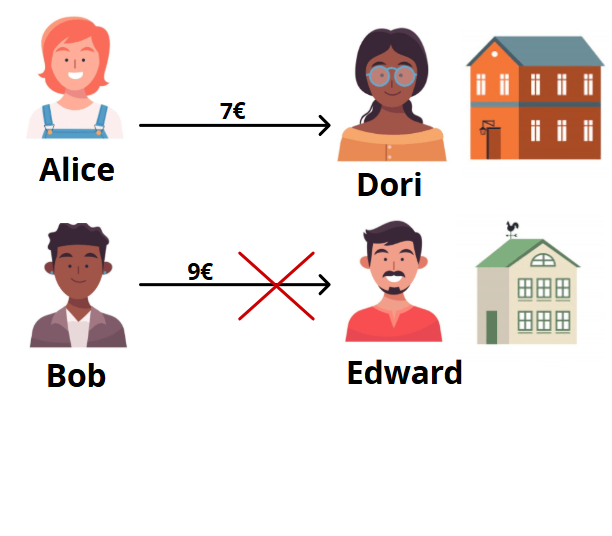}
			\caption{Individual Rationality}
			\label{fig:individual_rational}    
		\end{subfigure}
		\begin{subfigure}{0.49\textwidth}
			\centering
			\includegraphics[scale = 0.33]{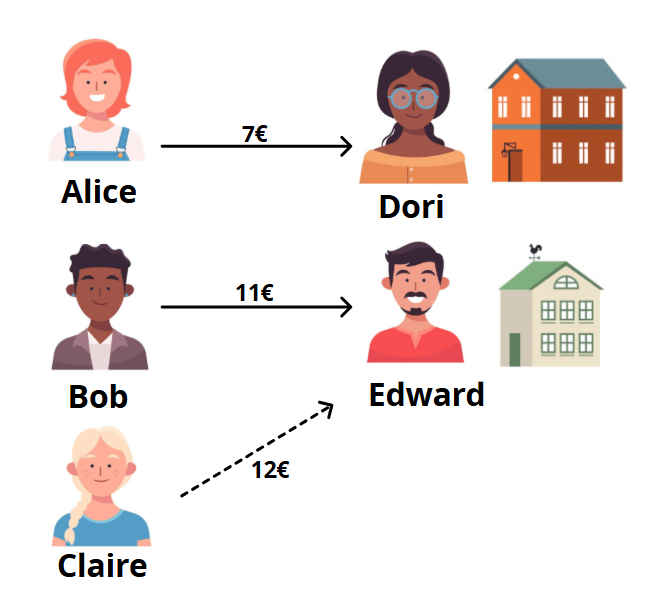}    
			\caption{Blocking Pair}
			\label{fig:blocking_pair}
		\end{subfigure}    
		\caption{Two possible allocations in our assignment game example. The first allocation is not individually rational, while the second allocation has a blocking pair.}
		\label{fig:stability_example}
	\end{figure}
	
	\begin{definition}
		Let $(\allocation)$ be an allocation. We define its \textbf{social welfare} as
		\begin{align*}
			\SW(\allocation) := \sum_{i\in B} u_i(\allocation) + \sum_{j\in S} v_j(\allocation).
		\end{align*}
		Notice that $ \SW(\allocation)$ does not depend on the price vector. In particular, we will alternatively denote it by $ \SW(\mu)$. A matching is called \textbf{optimal} if it maximizes the social welfare.
	\end{definition}
	
	For any pair $(i,j) \in B \times S$, define $a_{i,j} := h_{i,j}-c_j,$ which we refer to as their generated utility, and let $\boldsymbol{a} := (a_{i,j})_{i \in B, j \in S}$.
	We assume that $a_{i,j} \geq 0$ for all $(i,j) \in B \times S$. This assumption is made without loss of generality thanks to the individual rationality property, as agents generating negative utility will prefer to remain unmatched rather than matching together.
	
	\begin{remark}
		The notion of generated utility arises from the cooperative game perspective of the assignment game. In this setting, an assignment game can be modeled as pairs of agents forming matches, with each match producing a surplus. This surplus, or generated utility, is then divided between the two parties, reflecting the transferable-utility nature of the model. 
	\end{remark}

	
	Shapley and Shubik proved that stable allocations are necessarily based on an optimal matching. Optimal matchings can be obtained by solving a linear program, whose dual program outputs the respective agents utilities. The primal-dual linear programs considered by Shapley and Shubik are given below.
	
	\begin{minipage}{0.49\textwidth}
		\begin{align*}
			(P)\ &\max \sum_{i\in B}\sum_{j\in S} a_{i,j}\cdot x_{i,j}\\
			\text{s.t.}\ &\sum_{j\in S} x_{i,j} \leq 1, \forall i \in B,\\
			&\sum_{i\in B} x_{i,j} \leq 1, \forall j \in S,\\
			&x_{i,j} \in [0,1], \forall (i,j) \in B\times S.
		\end{align*}
	\end{minipage}
	\begin{minipage}{0.49\textwidth}
		\begin{align*}
			(D)\ &\min \sum_{i\in B} \alpha_i + \sum_{j\in S} \beta_j \\
			\text{s.t.}\ &\alpha_i + \beta_j  \geq a_{i,j}, \forall (i,j) \in B\times S,\\
			&\alpha_{i}, \beta_j \geq 0, \forall i \in B, \forall j \in S.\\
		\end{align*}
	\end{minipage}
	\medskip
	
	Note that $(P)$ always allows for integer optimal solutions thanks to the total-unimodularity of the matrix defined by the primal constraints. In particular, we will use interchangeably $\mu$ and $x$ when referring to matchings. 
	
	The utility vectors obtained as solutions of $(D)$, by construction, verify all conditions in Definition \ref{def:stability}. In terms of cooperative game theory, the utility vectors are said to belong to the core.
	
	\begin{definition}\label{def:core}
		Given an assignment game $\Gamma = (B,S,\boldsymbol{h},\boldsymbol{c})$ and its corresponding matrix of generated utility $\boldsymbol{a}$, we define its \textbf{core} as,
		\begin{align*}
			C(\Gamma) := \bigl\{(u,v) \in \mathbb{R}_+^{B}\times\mathbb{R}_+^{S} \mid u_i + v_j \geq a_{i,j}, \forall (i,j)\in B \times S\bigr\}.
		\end{align*}
	\end{definition}
	
	We state the main result of Shapley-Shubik's seminal article \cite{shapley1971assignment}.
	
	\begin{theorem}\label{teo:Shapley-Shubik}
		Let $\Gamma$ be an assignment game and $(x,\alpha,\beta)$ be solutions of the pair primal-dual linear programs such that $x$ is integral. It follows that $(\alpha,\beta)\in C(\Gamma)$ and there exists a price vector $\boldsymbol{p}$ such that $(x,\boldsymbol{p})$ is a stable allocation, with $u_i(x,\boldsymbol{p}) = \alpha_i$ and $v_j(x,\boldsymbol{p}) = \beta_j$, for any $i \in B, j\in S$.
		
		Conversely, let $(\allocation)$ be a stable allocation. Then, $\mu$ is an optimal matching and $$(u(\allocation), v(\allocation)) \in C(\Gamma).$$
	\end{theorem}
	
	Note that the prices vector $\boldsymbol{p}$ in the first part of \Cref{teo:Shapley-Shubik}'s statement are trivially given by,
	\begin{align*}
		p_j := \beta_j + c_j, \text{ for any } j \in S,
	\end{align*}
	as, by definition, the sellers' utility is given by the difference between the received payment and their valuation.
	
	\begin{example}
		In our running example, the optimal matching corresponds to matching Alice with Dori, Claire with Edward, and letting Bob unmatched. In addition, the optimal utilities can be defined for example by Alice paying $6$€ to Dori and Claire paying $12$€ to Edward.
	\end{example}
	
	

	\section{Sub-Optimality and Stability}\label{sec:optimality_and_stability}
	
	In real-world applications such as online markets, optimal matchings are rarely observed due to factors such as the sequential arrival of agents or the incomplete knowledge of their valuations. 
	As stated by Shapley-Shubik \cite{shapley1971assignment} (cf. \Cref{teo:Shapley-Shubik}), no allocation can be stable if its underlying matching is sub-optimal. 
	Hence, the purpose of this section is to measure the stability of sub-optimal matchings. To this end, we introduce two metrics: the stability index, based on the subset instability of \citet{jagadeesan2021learning}, and the approximate core \cite{faigle1998approximate,qiu2016approximate,vazirani2022general}. Each notion has its advantages and limitations. The stability index is comparatively easier to maximize, but it is sensitive to utility scaling and thus tends to focus on pairs that generate higher utility. The approximate core, while harder to achieve, avoids this bias and provides a more balanced measure of stability. 
	We formalize this intuition by proving that for any allocation, the distance to the core is bounded by the stability index. Moreover, we prove that both stability measures are bounded by the optimality ratio of the underlying matching, showing that by improving stability we indirectly build optimal allocations (cf. \Cref{eq:three_metrics}). First, we formalize the concept of sub-optimal matching.
	
	\begin{definition}
		Given an assignment game $\Gamma$ and an allocation $(\allocation)$, we define its \textbf{optimality ratio} as
		\begin{align*}
			\lambda(\allocation) = \lambda(\mu) := \frac{\SW(\mu)}{\opt},
		\end{align*}
		where $\opt$ denotes the social welfare of an optimal matching, that is, the solution of the primal problem $(P)$. 
	\end{definition}
	
	
	\subsection{Stability Index}
	
	The stability index is based on the concept of subset instability \cite{jagadeesan2021learning}, originally introduced as a notion of regret, that measures, for every sub-coalition, the difference between the social welfare it actually obtained and the maximum it could have obtained.
	
	\begin{definition}
		We define the \textbf{stability index} of a given allocation $(\allocation)$ as
		\begin{align*}
			\normSI(\allocation) := 1 - \frac{\SI(\allocation)}{\opt},
		\end{align*}
		with $\SI(\allocation)$ the subset instability of $(\allocation)$, formally defined as 
		\begin{align*}
			\SI(\mu,\mathbf{p}) := \max_{(B',S')\subseteq (B, S)} \max_{\mu' \in \mathcal{M}(B',S')}\{ \SW(\mu') -  \SW|_{B',S'}(\allocation)\},
		\end{align*}
		where $\mu'$ matches only agents from $B'$ to $S'$ and $ \SW|_{B',S'}(\mu, \mathbf{p})$ denotes the social welfare of $(\allocation)$ restricted to the sub-market $(B',S')$, that is,
		\begin{align*}
			\SW|_{B',S'}(\mu, \mathbf{p}) := \sum_{i\in B'} u_i(\allocation) + \sum_{j\in S'} v_j(\allocation).
		\end{align*}
	\end{definition}
	
	
	Note that subset instability does not require a price vector on the sub-market $(B',S')$ as prices cancel each other when computing the social welfare $\SW(\mu')$. However, since $(\allocation)$ is an allocation on the whole market $(B, S)$, the prices do not necessarily cancel out in $ \SW|_{B',S'}(\mu, \mathbf{p})$.
	
	\begin{example}\label{ex:subset_instability}
		Consider the allocation $(\allocation)$ illustrated in \Cref{fig:subset_inst}. It follows that $\lambda(\allocation) = \frac{2}{3}$ as both couples generate, respectively, $4$€ and $2$€, while the optimal matching has social welfare equal to $9$€. Regarding the stability index, note that considering the submarket defined by Alice, Claire, Dori, and Edward, and matching Alice with Dori and Claire with Edward, the difference on social welfare is equal to $4$€. Since no other combination of submarket and matching creates a higher difference, $\SI(\allocation) = 4$€, and, thus, $\normSI(\allocation) = \frac{5}{9}$. 
		\begin{figure}[H]
			\centering
			\includegraphics[scale = 0.3]{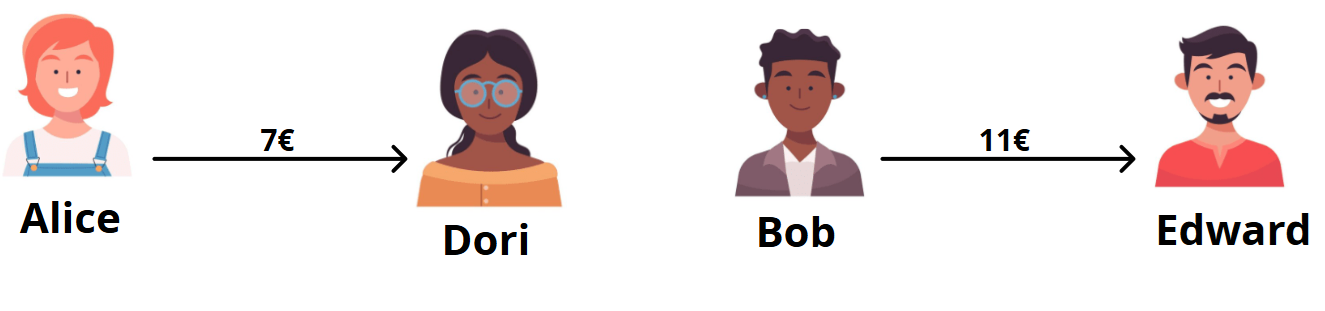}
			\caption{An allocation with $\normSI(\allocation) = \frac{5}{9}$.}
			\label{fig:subset_inst}
		\end{figure}
	\end{example}
	
	\citet{jagadeesan2021learning} showed that the subset instability of a given allocation is always at least as large as the additive optimality gap of its matching. Consequently, the stability index is always upper-bounded by the optimality ratio, as illustrated in Example~\ref{ex:subset_instability}. As a complementary result, we show that by carefully choosing the prices of an allocation, the stability index can be maximized to exactly match the optimality ratio. To this end, we consider the notion of a stabilizing subsidy and recall a technical result, both drawn from \citet{jagadeesan2021learning}.
	
	
	\begin{definition}\label{def:minimum_stab_subsidy}
		Let $(\allocation)$ be an allocation. We define the \textbf{minimum stabilizing subsidy} as the solution of the following problem,
		\begin{align*}
			\min_{(\boldsymbol{\tau},\boldsymbol{\eta})\in\mathbb{R}_+^{B}\times\mathbb{R}_+^{S}} &\sum_{i\in B} \tau_i + \sum_{j\in S} \eta_j & \\
			\text{s.t.}\ &u_i(\allocation) + \tau_i \geq 0, &\forall i \in B,\\
			&v_j(\allocation) + \eta_j \geq 0,  &\forall j \in S,\\
			&u_i(\allocation) + \tau_i + v_j(\allocation) + \eta_j \geq a_{i,j}, &\forall (i,j) \in B \times S.
		\end{align*}
	\end{definition}
	
	A minimum stabilizing subsidy corresponds to the minimum utility injection required to make the allocation $(\allocation)$ stable.
	
	\begin{lemma}\label{lemmma:subset_inst_and_stab_subsidy}
		For any allocation, the minimum stabilizing subsidy is equal to the subset instability.
	\end{lemma}
	
	We are now able to show the following result.
	
	\begin{theorem}\label{thm:opt_gap_and_subset_inst}
		For any allocation $(\allocation)$, it always holds $\normSI(\allocation)\leq\lambda(\mu)$. Moreover, for any matching $\mu$, there exists $\boldsymbol{p}\in\mathbb{R}_+^{S}$ such that, $(\allocation)$ is individually rational and $\normSI(\allocation) = \lambda(\mu)$.
	\end{theorem}
	
	\begin{proof}
		The upper bound for the stability index is a consequence of the results of \citet{jagadeesan2021learning}. We focus then on proving the second part of the statement. As observed by \citet{shi2023optimal}, without loss of generality, we can assume that $|B|=|S|$ and that $\mu$ is a maximum size matching, i.e, no agent in $B\cup S$ is unmatched. By Lemma \ref{lemmma:subset_inst_and_stab_subsidy}, minimizing the subset instability over the prices vector can be written as
		\begin{align*}
			&\hspace{-1.5cm}\min_{(\boldsymbol{p},\boldsymbol{\tau},\boldsymbol{\eta})\in\mathbb{R}_+^{S}\times\mathbb{R}_+^{B}\times\mathbb{R}_+^{S}} \sum_{i\in B} \tau_i + \sum_{j\in S} \eta_j &\\
			\hspace{1cm}\text{s.t.}\ &h_{i,\mu_i} - p_{\mu_i} + \tau_i \geq 0, &\forall i \in B,\\
			&p_j - c_j + \eta_j \geq 0, &\forall j \in S,\\
			&h_{i,\mu_i} - p_{\mu_i} + \tau_i + p_j - c_j + \eta_j \geq a_{i,j}, &\forall (i,j) \in B \times S.
		\end{align*}
		The dual of the previous problem is (see Appendix \ref{sec:missing proofs}),
		\begin{align*}
			\max_{\boldsymbol{\alpha},\boldsymbol{\beta},\boldsymbol{\gamma} \geq 0} &\sum_{i\in B}\sum_{j\in S} \gamma_{i,j}(h_{i,j} - h_{i,\mu_i}) + \sum_{j\in S}\beta_j c_j  -\sum_{i\in B}\alpha_i h_{i,\mu_i} \\
			\text{s.t.}\ & \alpha_i + \sum_{j\in S}\gamma_{i,j} = 1 , \forall i \in B,\\
			&\beta_j + \sum_{i\in B}\gamma_{i,j} \leq 1, \forall j \in S.
		\end{align*}
		Using the first constraint, note that the objective function becomes:
		$$ \sum_{i\in B}\sum_{j\in S} \gamma_{i,j}h_{i,j} + \sum_{j\in S}\beta_j c_j  -\sum_{i\in B}h_{i,\mu_i}.$$
		
		
			
		
		\noindent The coefficients of $\boldsymbol{\beta}$ in the objective function being positive, the value of the problem is not modified by replacing the second class of constraints by
		$$ \beta_j + \sum_{i\in B}\gamma_{i,j} = 1, \forall j \in S.$$
		Using this, the objective function becomes
		$$ \sum_{i\in B}\sum_{j\in S} \gamma_{i,j}(h_{i,j} - c_j) + \sum_{j\in S} c_j  - \sum_{i\in B}h_{i,\mu_i}. $$
		Since none of the variables within $\boldsymbol{\alpha}$ and $\boldsymbol{\beta}$ appear in the objective function, they correspond to slack variables. Furthermore, as all agents are matched, it holds,
		\begin{align*}
			&\sum_{i\in B}h_{i,\mu_i} - \sum_{j\in S} c_j = \sum_{(i,j)\in \mu}h_{i,j} - c_j =  \SW(\allocation). 
		\end{align*}
		Thus, the problem is finally written as,
		\begin{align*}
			\max & \sum_{i\in B}\sum_{j\in S} \gamma_{i,j}(h_{i,j}-c_j) - \SW(\mu) \\
			\text{s.t.}\ & \sum_{j\in S}\gamma_{i,j} \leq 1 , \forall i \in B,\\
			& \sum_{i\in B}\gamma_{i,j} \leq 1, \forall j \in S,\\
			&\boldsymbol{\beta},\boldsymbol{\gamma} \geq 0. 
		\end{align*}
		The resulting problem corresponding to problem (P) shifted of $ - \SW(\mu)$, we obtain that, $$\min_{\boldsymbol{p}}\SI(\mu, \boldsymbol{p}) = \max_{\mu' \in \matching{B}{S}} \SW(\mu') - \SW(\mu) = \opt - \SW(\mu).$$
		Regarding the choice of prices such that the corresponding allocation is individually rational, please refer to Appendix \ref{sec:missing proofs}.
		\qedhere
	\end{proof}
	
	As shown by \Cref{thm:opt_gap_and_subset_inst}, whenever prices can be chosen after the matching has been computed, the stability index can be maximized to exactly match the optimality ratio. While this represents an important theoretical bound, in most of online matching applications, prices must be decided at the same time that couples are matched, in particular, ignoring the future pairs to be created. The following result provides a simple method to ensure a $\nicefrac{1}{2}$-guarantee in such cases.
	
	\begin{proposition}\label{prop:half_prices_is_half_stable}
		Let $\mu$ be a matching. Consider the price vector $\boldsymbol{p}^{\it half} := (p_j^{\it half})_{j\in S} \in \mathbb{R}_+^S$, defined by,
		\begin{align*}
			p_j^{\it half} := \left\{ 
			\begin{array}{cc}
				c_j + \frac{1}{2}\cdot a_{i,j} & \text{ if } \mu_j = i,\\
				c_j     & \text{ if } \mu_j = j.
			\end{array}
			\right.    
		\end{align*}
		It follows that $\frac{1}{2}\cdot \lambda(\mu) \leq \normSI(\mu, \boldsymbol{p}^{\it half})$.
	\end{proposition}
	
	
	
	\begin{proof}
		Let $\mu$ and $\boldsymbol{p}^{\it half}$ be as stated. Denote  $\boldsymbol{p}^* \in \argmin_{\boldsymbol{p}}\SI(\allocation)$ such that $(\mu,\boldsymbol{p}^*)$ is individually rational (\Cref{thm:opt_gap_and_subset_inst}). 

		For $i \in B$, notice that $u_i(\mu,\boldsymbol{p}^{\it half}) = \frac{1}{2}\cdot a_{i,\mu_i}$ and, since $(\mu,\boldsymbol{p}^*)$ is individually rational,  $u_i(\mu,\boldsymbol{p}^*) \in [0, a_{i, \mu_i}]$. Thus, $u_i(\mu,\boldsymbol{p}^{\it half}) \geq \frac{1}{2} u_i(\mu,\boldsymbol{p}^*)$. Similarly, for $j \in S$, it also holds that $v_j(\mu,\boldsymbol{p}^{\it half}) \geq \frac{1}{2}\cdot v_j(\mu,\boldsymbol{p}^*)$. Given $(B', S') \subseteq (B,S)$ and $\mu'\in\mathcal{M}(B',S')$, it follows,
		
		\begin{align*}
			\SW(\mu') -  \SW|_{B',S'}(\mu,\boldsymbol{p}^{\it half}) &=  \SW(\mu') - \sum_{i\in B'} u_i(\mu,\boldsymbol{p}^{\it half})- \sum_{j\in S'} v_j(\mu,\boldsymbol{p}^{\it half})\\
			&\leq \SW(\mu') - \frac{1}{2}\cdot \sum_{i\in B'}u_i(\mu,\boldsymbol{p}^*) - \frac{1}{2} \cdot \sum_{j\in S'}v_j(\mu,\boldsymbol{p}^*)\\
			&= \SW(\mu') - \frac{1}{2}\cdot \SW|_{B',S'}(\mu,\boldsymbol{p}^{*}) \\
			&= \SW(\mu') - \SW|_{B',S'}(\mu,\boldsymbol{p}^{*}) + \frac{1}{2}\cdot \SW|_{B',S'}(\mu,\boldsymbol{p}^{*})\\
			&\leq \SI(\mu,\boldsymbol{p}^*) + \frac{1}{2}\cdot \SW|_{B',S'}(\mu,\boldsymbol{p}^{*})\\
			&\leq \opt - \SW(\mu) + \frac{1}{2}\cdot \SW(\mu) \\
			&= \opt - \frac{1}{2}\cdot \SW(\mu),
		\end{align*}
		where the last inequality uses the fact that $\SI(\mu,\boldsymbol{p}^*) = \opt - \SW(\mu)$. We conclude by taking the maximum over all sub-markets and matchings, and normalizing by $\opt$.
	\end{proof}
	
	Observe that prices in Proposition \ref{prop:half_prices_is_half_stable} only depend on the generated utility of the seller and buyer, making it computable online settings where agents or edges arrive sequentially to the market.
	
	\subsection{Approximated Core}
	
	The stability index is not scale-invariant with respect to the agents' utility. For example, consider the market in \Cref{fig:housing_market_example} and change $h_{\text{Alice},\text{Dori}}$ to $10^{10}$. Matching only Alice and Dori, with a price equal to $6$€, achieves a stability index close to $1$, while leaving unmatched the rest of the agents. To avoid this, we consider an alternative stability notion, known as the approximate core \cite{faigle1998approximate,qiu2016approximate, vazirani2022general}.
	
	
	\begin{definition}
		Given an assignment game $\Gamma$, its corresponding matrix of generated utility $\boldsymbol{a}$, and $\kappa\in [0,1]$, we define the \textbf{$\kappa$-approximate core} as
		\begin{align*}
			C_\kappa(\Gamma) := \bigl\{(u,v) \in \mathbb{R}_+^{B}\times\mathbb{R}_+^{S} \mid u_i + v_j \geq \kappa\cdot a_{i,j}, \forall (i,j)\in B \times S\bigr\}.
		\end{align*}
		We denote $(\allocation) \in C_\kappa(\Gamma)$ whenever $(u(\allocation),v(\allocation))$ belongs to $C_\kappa(\Gamma)$, and say that $(\allocation)$ is in the $\kappa$-approximate core.
	\end{definition}
	
	Matching only Alice and Dori in the modified market with $h_{\text{Alice},\text{Dori}} = 10^{10}$, even though it achieves a stability index close to $1$, it belongs to the $0$-approximate core. With this in mind, as a first result, we prove that belonging to the $\kappa$-approximate core is indeed stronger than achieving a stability index of $\kappa$.
	
	\begin{proposition}\label{prop:connection_approx_core_and_subset_inst}
		Let $(\allocation)$ be an allocation in the $\kappa$-approximate core. Then, $\kappa\leq\normSI(\allocation)$.
	\end{proposition}
	
	\begin{proof}
		Let $(\allocation)$ be a $\kappa$-approximate core allocation, $(B', S') \subseteq (B,S)$ a sub-market, and $\mu' \in \mathcal{M}(B',S')$ a matching. It follows,
		\begin{align*}
			\SW(\mu') -  \SW|_{B',S'}(\allocation) &= \sum_{(i,j)\in\mu'} a_{i,j} - \sum_{i \in B'}u_i(\allocation) - \sum_{j\in S'}v_j(\allocation)\\
			&= \sum_{(i,j)\in\mu'} a_{i,j} - u_i(\allocation) - v_j(\allocation) - \sum_{(i,i)\in \mu'} u_i(\allocation) -\sum_{(j,j) \in \mu'}  v_j(\allocation)\\
			&\leq \sum_{(i,j)\in\mu'} \bigl(a_{i,j} - u_i(\allocation) - v_j(\allocation)\bigr)  \\
			&\leq \sum_{(i,j)\in\mu'} a_{i,j} - \kappa \cdot a_{i,j} = (1-\kappa) \SW(\mu')\\
			&\leq (1-\kappa)\opt,
		\end{align*}
		where the first inequality comes from individual rationality (utilities are non-negative in the $\kappa$-approximate core) and the second one from $(\allocation)$ being in the $\kappa$-approximate core. We conclude by taking maximum over all sub-markets and all matchings.
	\end{proof}
	
	
	The $\kappa$-approximate core is a \textit{local} stability notion, as it evaluates the social welfare of each pair relative to their generated utility. In contrast, subset instability is a \textit{global} stability notion, since it considers the aggregated social welfare across all pairs. Intuitively, if an allocation is locally close to being stable everywhere, then it must also be globally stable (Proposition \ref{prop:connection_approx_core_and_subset_inst}), whereas the converse does not necessarily hold, as we show next.
	
	\begin{proposition}\label{prop:counterexample_0_kappa_stable}
		For any $\kappa \in [0,1)$, there exists an assignment game with an allocation $(\allocation)$ verifying $\kappa \leq \normSI(\allocation)\leq \lambda(\mu)$, such that for no constant $\kappa' \in (0, \kappa]$, $(\allocation)$ is in the $\kappa'$-approximate core.
	\end{proposition}
	\begin{proof}
		Let $\kappa\in [0,1)$ be a constant. Consider an assignment game with two buyers $B = \{a,b\}$, two sellers $S = \{\alpha,\beta\}$, and the following matrix of generated utility $\boldsymbol{a}$
		\begin{align*}
			\boldsymbol{a} = 
			\begin{array}{c|cc}
				& \alpha & \beta \\
				\hline
				a &  \kappa  & 0  \\
				b & 0   & 1-\kappa
			\end{array}
		\end{align*}
		Consider $(\allocation)$ defined by $\mu = \{(a,\alpha), (b,b),(\beta,\beta)\}$ and $\boldsymbol{p} = (0,0)$, that is, only $a$ and $\alpha$ are matched and $a$ pays $0$ to $\alpha$. The allocation verifies $\opt -  \SW(\allocation) \leq \SI(\allocation)\leq (1-\kappa)\opt$, however, $(\allocation)$ is not in the $\kappa'$-approximate stable for any $\kappa' > 0$.
	\end{proof}
	
	Interestingly, there exists a connection between the $\kappa$-approximate core and a multiplicative version of subset instability.
	
	\begin{theorem}\label{thm:allocations_are_always_kappa_core}
		Given $(\allocation)$ an individually rational allocation, define,
		\begin{align}\label{eq:alt_kappa}
			\kappa(\allocation) := \min\limits_{(i,j) \in B\times S} \frac{1}{a_{i,j}}\cdot(u_i(\allocation) + v_j(\allocation)).
		\end{align}
		Then, it always holds that $(\allocation)$ is in the $\kappa(\allocation)$-approximated core. In addition, 
		\begin{align}\label{eq:max_kappa}
			\kappa(\allocation) = \min\limits_{(B',S') \subseteq (B,S)} \min_{\mu'\in\mathcal{M}(B',S')} \dfrac{\SW|_{B',S'}(\mu)}{\SW(\mu')}.
		\end{align}
	\end{theorem}
	
	\begin{proof}
		Let $(\allocation)$ be an individually rational allocation. Recall that $(\allocation)$ is in the $\kappa$-approximate core, for $\kappa$ some constant, if
		\begin{align*}
			\forall (i,j) \in B\times S, u_i(\allocation) + v_j(\allocation) \geq \kappa \cdot a_{i,j}
			\Longleftrightarrow &\ \forall (i,j) \in B\times S, \frac{u_i(\allocation) + v_j(\allocation)}{a_{i,j}} \geq \kappa \\
			\Longleftrightarrow &\ \min\limits_{(i,j) \in B\times S} \frac{1}{a_{i,j}}\cdot(u_i(\allocation) + v_j(\allocation)) \geq \kappa.
		\end{align*}
		Therefore, $(\allocation)$ always belongs to the $\kappa(\allocation)$-approximated core, for $\kappa(\allocation)$ as in \Cref{eq:alt_kappa}. We prove next that \Cref{eq:max_kappa} holds. Consider 
		$$R := \min\limits_{(B',S') \subseteq (B,S)} \min_{\mu'\in\mathcal{M}(B',S')} \dfrac{\SW|_{B',S'}(\mu)}{\SW(\mu')}.$$
		It directly follows that $\kappa(\allocation) \geq R$ as $R$ considers all sub-markets, in particular those with only one agent per side. Consider next $(B',S') \subseteq (B,S)$ and $\mu'$ a matching from $B'$ to $S'$. Consider, without loss of generality, that $|B'|=|S'|$ and all agents are matched at $\mu'$ (indeed, removing any unmatched agent from the coalition does not affect $\SW(\mu')$ and does not decrease $\SW|_{B',S'}(\mu)$, by individual rationality). It follows,
		\begin{align*}
			\frac{\SW|_{B',S'}(\mu)}{\SW(\mu')} &= \frac{\sum\limits_{(i,j)\in\mu'} u_i(\allocation) + v_j(\allocation)}{\sum\limits_{(i,j)\in \mu'} a_{i,j}} \geq \min_{(i,j)\in\mu'} \frac{u_i(\allocation) + v_j(\allocation)}{a_{i,j}} \geq \kappa(\allocation).
		\end{align*}
		$\kappa(\allocation)$ not depending on $(B',S')$ nor $\mu'$, we conclude $\kappa(\allocation)\leq R.$
	\end{proof}
	
	\Cref{thm:allocations_are_always_kappa_core} shows that an allocation will be as unstable as its most unstable couple. In the proof of Proposition \ref{prop:counterexample_0_kappa_stable}, for example, the value of $\kappa(\allocation)$ of the constructed allocation is equal to $0$. 
	
	To conclude the section, putting together Proposition \ref{prop:connection_approx_core_and_subset_inst} and \Cref{thm:allocations_are_always_kappa_core}, we conclude that for any allocation $(\allocation)$, it holds
	\begin{align}\label{eq:three_metrics}
		\kappa(\allocation) \leq \mathcal{J}(\allocation) \leq \lambda(\mu).
	\end{align}
	
	\Cref{eq:three_metrics} is particularly significant when applied to online matching, as it suggests that algorithms focusing on obtaining good stability bounds will invariably obtain good optimality bounds.
	
	\section{Online Stable Allocations}\label{sec:online_stable_matching}
	
	This section considers randomized algorithms to find stable allocations in online assignment games. After adapting our stability metrics to uncertain settings, we obtain systematic bounds on optimality and stability in two well-known models of online matching.
	
	\subsection{Stability Under Uncertainty}
	
	We consider two standard online matching frameworks: the \textit{edge arrival} model (\Cref{fig:edge_arrival}) and the \textit{vertex arrival} model (\Cref{fig:vertex_arrival}). In the former, we start with a bipartite graph containing only vertices, with edges arriving one by one. Upon the arrival of an edge, the algorithm must irrevocably decide whether to accept it, specifying a price to be paid, or to reject it. In the latter, we start with a bipartite graph with vertices fixed on one side, while vertices on the other side (together with their incident edges) arrive sequentially. Upon the arrival of such a vertex (an agent), the algorithm must decide whether to match it to an available partner (possibly none) and, if matched, at what price. Whenever randomization is allowed in these decisions, we refer to the algorithm as randomized.
	\begin{figure}[H]
		\begin{subfigure}{0.49\textwidth}
			\centering
			\includegraphics[scale = 0.33]{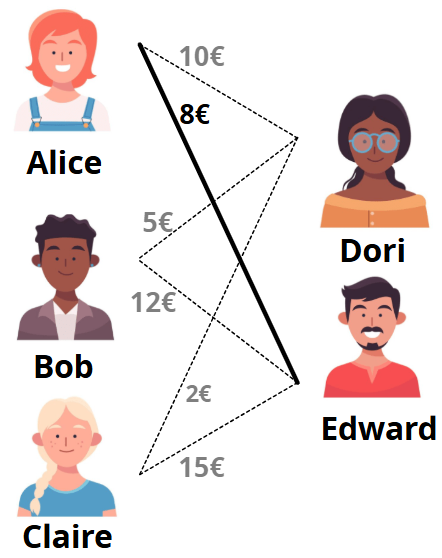}
			\caption{Edge arrival model}
			\label{fig:edge_arrival}    
		\end{subfigure}\hfill
		\begin{subfigure}{0.49\textwidth}
			\centering
			\includegraphics[scale = 0.33]{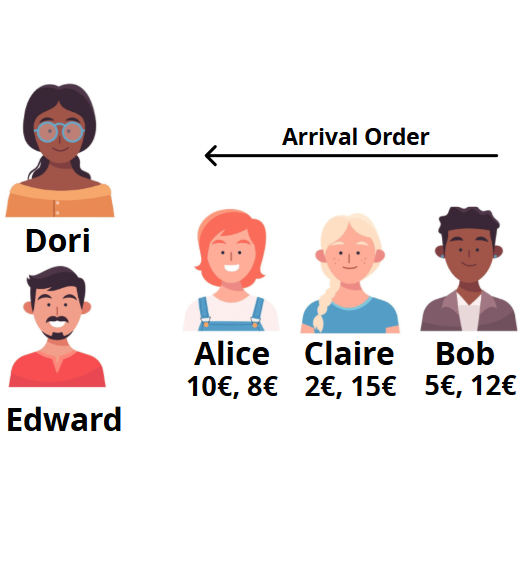}    
			\caption{Vertex arrival model}
			\label{fig:vertex_arrival}
		\end{subfigure}    
		\caption{Two online matching markets models.}
		\label{fig:online_matching_markets}
	\end{figure}
	Given an online assignment game instance $\Gamma = (B,S, \boldsymbol{h},\boldsymbol{c})$ (either edge or vertex arrival)\footnote{For simplicity, we keep $\Gamma$ to denote online assignment games.} and a randomized algorithm $\ALG$, we denote $\ALG(\Gamma)$ the probability distribution of outcomes generated by $\ALG$, $supp(\ALG(\Gamma))$ its support, that is, the set of possible outcomes of the algorithm on $\Gamma$, and $(u(\ALG, \Gamma),v(\ALG,\Gamma))$ the vectors of expected utilities of the agents, that is, for any $(i,j)\in B\times S$,
	\begin{align*}
		&u_i(\ALG,\Gamma) = \mathbb{E}_{(\allocation)\sim \ALG(\Gamma)}[u_i(\allocation)] \text{ and } v_j(\ALG,\Gamma) = \mathbb{E}_{(\allocation)\sim \ALG(\Gamma)}[v_j(\allocation)],
	\end{align*}
	where $(\allocation)\sim \ALG(\Gamma)$ indicates that allocations are sampled from the distribution induced by ALG in $\Gamma$.  Given a metric $m \in \{\lambda,\normSI,\kappa\}$, we define the corresponding ex-post and ex-ante metric as follows,
	\begin{align*}
		&m^{\text{post}}(\ALG,\Gamma) := \min_{(\allocation)\in supp(\ALG(\Gamma))} m(\allocation) \text{ and }
		m^{\text{ante}}(\ALG,\Gamma) := \mathbb{E}_{(\allocation)\sim \ALG(\Gamma)}[m(\allocation)].
	\end{align*}
	Additionally, we define the average metric $m^{\text{avg}}(\ALG,\Gamma)$ by applying the definition of a metric $m$ on the expected utilities of the agents, that is, on $(u(\ALG,\Gamma),v(\ALG,\Gamma))$.
	
	\begin{remark}
		In both the online matching literature and the fair division literature, the ex-ante and average guarantees are usually treated as equivalent (see, e.g., \cite{freeman2020best,immorlica2023online}). Indeed, whenever a metric is linear, its ex-ante and average version coincide. 
		In particular, it always holds that $\lambda^{\text{ante}}(\ALG,\Gamma) = \lambda^{\text{avg}}(\ALG,\Gamma)$. However, the non-linearity of our stability metrics $\normSI$ and $\kappa$ breaks this equivalence, motivating the two different definitions in the randomized case (ex-ante and average), and showing that achieving stability in the online assignment game is a more subtle problem than achieving optimality.
	\end{remark}
	
	We now show a useful result that allows us to systematize stability and optimality guarantees of randomized online algorithms:
	
	\begin{proposition}\label{prop:post_bounds_ante_bounds_worst}
		Let $\Gamma$ be an instance, $\ALG$ a randomized algorithm, and $m \in \{\lambda,\normSI,\kappa\}$ a metric. It always holds,
		\begin{align*}
			m^{\text{post}}(\ALG,\Gamma) \leq m^{\text{ante}}(\ALG,\Gamma) \leq m^{\text{avg}}(\ALG,\Gamma).    
		\end{align*}
		Additionally, for any $\gamma\in \{\text{post},\text{ante},\text{avg}\}$, it always holds,
		\begin{align*}
			\kappa^{\gamma}(\ALG,\Gamma) \leq \normSI^{\gamma}(\ALG,\Gamma) \leq \lambda^{\gamma}(\ALG,\Gamma).
		\end{align*}
	\end{proposition}
	
	The proof of the first part of Proposition \ref{prop:post_bounds_ante_bounds_worst} uses the linearity of the expected value and Jensen's inequality. The second part adapts the arguments on the proof of Proposition \ref{prop:connection_approx_core_and_subset_inst}. The formal proof is included in Appendix \ref{sec:missing proofs}. 
	
	
	\subsection{Equal Pricing and Edge Arrival Model}
	
	We begin by stating a generalization of Proposition \ref{prop:half_prices_is_half_stable} and by strengthening it to a tightness result that shows the impossibility of obtaining good stability guarantees under the edge arrival model.
	
	\begin{proposition}\label{prop:randomized_half_prices}
		Let $\Gamma$ be an instance and $\ALG$ a randomized matching algorithm. Denote $\ALG + \textit{Half}$ the randomized algorithm that matches agents following $\ALG$ and, for any pair of matched agents $(i,j) \in B\times S$, the price $p_j$ is set as in Proposition \ref{prop:half_prices_is_half_stable}. For any $\gamma\in \{\text{post},\text{ante},\text{avg}\}$, it holds,
		$$\frac{1}{2}\cdot\lambda^{\gamma}(\ALG + \textit{Half},\Gamma) \leq \mathcal{J}^{\gamma}(\ALG +  \textit{Half},\Gamma).$$
	\end{proposition}
	
	The proof of Proposition \ref{prop:randomized_half_prices} follows similar arguments than Proposition \ref{prop:half_prices_is_half_stable} and can be found in Appendix \ref{sec:missing proofs}.
	
	The pricing algorithm \textit{Half} considered in Proposition \ref{prop:randomized_half_prices} may incorporate any randomized matching procedure and is applicable to both the edge and vertex arrival models. Interestingly, in the edge arrival setting we can construct two simple instances such that no algorithm from a broad family of randomized algorithms can achieve better than the $\nicefrac{1}{2}$-factor on both instances simultaneously.
	
	\begin{proposition}\label{prop:tightness_half_prices}
		Consider the two edge arrival instances $\Gamma_1$ and $\Gamma_2$ illustrated in Figure \ref{fig:example_tightness}, where the first edge to arrive in each of them is between Alice and Dori, and all generated utilities are equal to $1$.
		\begin{figure}[H]
			\centering
			\includegraphics[scale = 0.225]{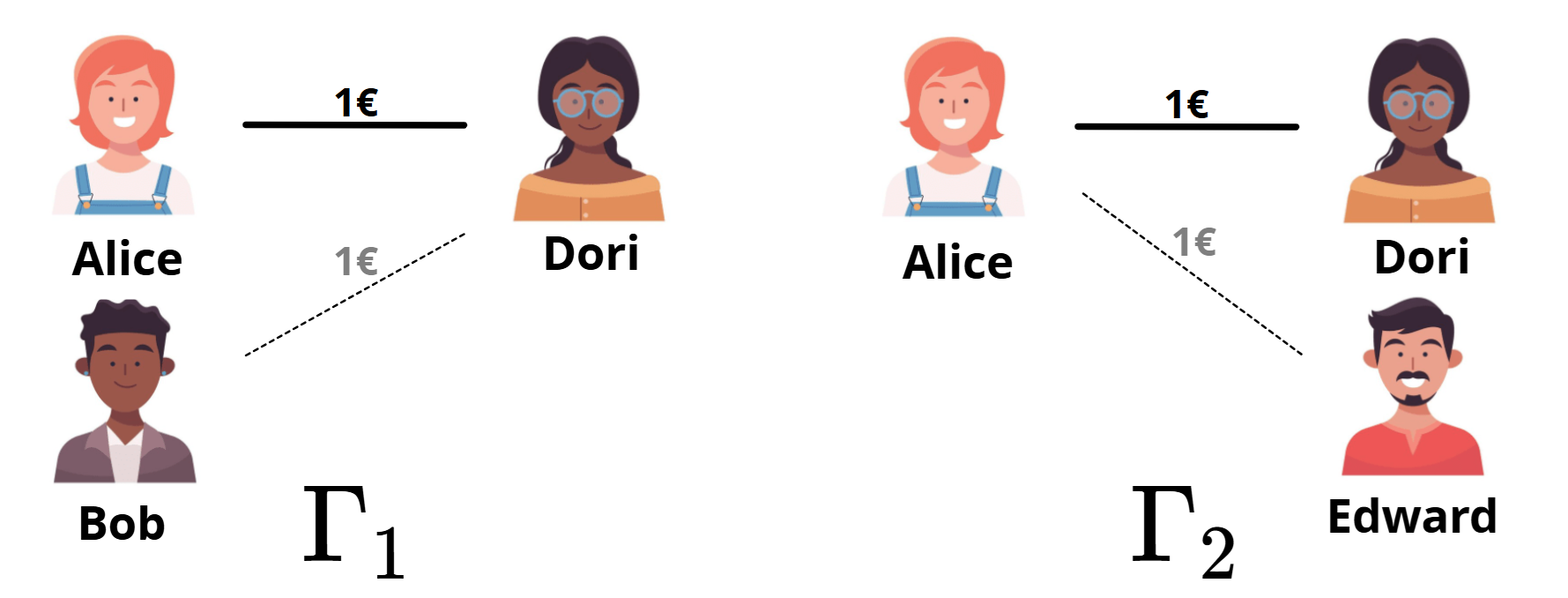}
			\caption{Two edge arrival instances}
			\label{fig:example_tightness}
		\end{figure}
		Let $\mathcal{A}$ be the family of all randomized algorithms such that Alice and Dori are matched with probability $1$. Then, for any $\ALG \in \mathcal{A}$ and any $\gamma\in \{\text{post},\text{ante},\text{avg}\}$, it holds,
		$$\frac{1}{2}\cdot\lambda^{\gamma}(\ALG,\Gamma_1) \geq \mathcal{J}^{\gamma}(\ALG,\Gamma_1) \text{ or } \frac{1}{2}\cdot\lambda^{\gamma}(\ALG,\Gamma_2) \geq \mathcal{J}^{\gamma}(\ALG,\Gamma_2).$$
	\end{proposition}
	\begin{proof}
		From Proposition \ref{prop:post_bounds_ante_bounds_worst}, given $\ALG\in\mathcal{A}$, it is enough to prove that either $\frac{1}{2}\cdot\lambda^{\text{post}}(\ALG,\Gamma_1) \geq \mathcal{J}^{\text{avg}}(\ALG,\Gamma_1)$ or $\frac{1}{2}\cdot\lambda^{\text{post}}(\ALG,\Gamma_2) \geq \mathcal{J}^{\text{avg}}(\ALG,\Gamma_2)$. Let $\ALG\in\mathcal{A}$ be a randomized algorithm. Suppose, once $\ALG$ matches Alice and Dori, it holds $\mathbb{E}[u_{\text{Alice}}(\ALG)] \geq \mathbb{E}[v_{\text{Dori}}(\ALG)]$. Notice that this assumption is independent on the chosen instance. Choose the first instance, where Bob was already present on the market and the second (and last) edge to arrive is between him and Dori. Since Dori is already matched, the arriving edge is wasted. In this case, therefore, $\lambda^{\text{post}}(\ALG,\Gamma_1) = \frac{1}{1} = 1$ while
		\begin{align*}
			\SI^{\text{avg}}(\ALG,\Gamma_1) = a_{\text{Bob}, \text{Dori}} - \mathbb{E}[v_{\text{Dori}}(\ALG)] \geq \frac{1}{2},
		\end{align*}
		since $\mathbb{E}[u_{\text{Alice}}(\ALG)] + \mathbb{E}[v_{\text{Dori}}(\ALG)] = 1$. We conclude, 
		\begin{align*}
			\normSI^{\text{avg}}(\ALG,\Gamma_1) = 1 - \frac{\SI^{\text{avg}}(\ALG,\Gamma_1)}{\opt} \leq \frac{1}{2} = \frac{1}{2}\cdot\lambda^{\text{post}}(\ALG,\Gamma_1).
		\end{align*}
		The proof if $\mathbb{E}[u_{\text{Alice}}(\ALG)] \leq \mathbb{E}[v_{\text{Dori}}(\ALG)]$ is analogous considering the second edge arrival instance.
	\end{proof}
	\vspace{-0.12cm}
	
	Proposition \ref{prop:tightness_half_prices} highlights the difficulty of obtaining good stability bounds in edge arrival models when studying stable matchings with transferable utility. As exposed by \citet{rochford_symmetrically_1984}, achieving stability requires respecting the agents’ threat levels, i.e., the best utility each agent can secure outside their current match. In the edge arrival setting, however, these threat levels evolve dynamically on both sides of the market, which makes stability particularly challenging to maintain. This problem is only \textit{partially} present in the vertex arrival model, as the threat levels of the buyers (the arriving agents) is much more determined at the moment of arrival.

	\subsection{Vertex Arrival Model}
	
	
	In this section, we distinguish between two variants of the vertex arrival model: vertex-weighted and edge-weighted with free disposal. In the vertex-weighted model, the generated utilities do not depend on the identity of the buyer; that is, for any $(i,i',j) \in B \times B \times S$, whenever both $a_{i,j}$ and $a_{i',j}$ are not zero, then $a_{i,j} = a_{i',j}$.
	
	The edge-weighted model, by contrast, captures the classical situation in which different buyer–seller pairs may generate different utilities. In addition, the free-disposal assumption states that when a new buyer arrives, the decision-maker may unmatch a previously formed pair in order to reassign the seller to the arriving buyer, thereby leaving the previously matched buyer unmatched. 
	
	In each of these models, previous works have studied the design of \emph{competitive} algorithms (see \cite{echenique_online_2023} for a complete survey). Informally, a randomized algorithm is competitive if it achieves a constant factor for the optimality ratio over all instances (the constant not depending on the instance). In our notation, an algorithm is competitive if $\lambda^{post}$ (for deterministic algorithms) and $\lambda^{avg}$ (for randomized algorithms) can be lower bounded by a constant not depending on the instance. Therefore, we can leverage the literature results to obtain stability guarantees for randomized algorithms.
	
	\medskip
	
	\noindent\textbf{Vertex-weighted}. In the vertex-weighted setting, \citet{karp1990optimal} introduced the \emph{Ranking} algorithm which, in its alternative version (see Algorithm \ref{alg:ranking} in Appendix \ref{sec:missing proofs}), it randomly sets prices to goods before the arrival of any agent, and upon arrival of a buyer, it matches it greedily. \citet{aggarwal2011online} proved that Ranking is optimal, that is, it achieves an average competitive ratio of $1-\frac{1}{e}$. A closer examination of the proof of \textit{Ranking}'s optimality shows that it defines transaction prices inducing expected utilities that belong to the $\kappa$-approximated core.  
	In particular, we are able to restate the result of \citet{aggarwal2011online} in our terminology:
	
	\begin{proposition}\label{prop:vertex_weighted_kappa_post_is_opt}
		Consider the Ranking algorithm in the vertex-weighted setting. It holds,
		\begin{align*}
			\min_{\text{Instance } \Gamma} \kappa^{\text{avg}}(\text{Ranking},\Gamma) = 1 - \frac{1}{e}.
		\end{align*}
	\end{proposition}
	
	A more recent version of the proof of Ranking's optimality (Theorem 5.6  in \cite{echenique_online_2023}) is based on a technical result (Lemma 5.5 in \cite{echenique_online_2023}), which corresponds exactly to Proposition \ref{prop:vertex_weighted_kappa_post_is_opt}. Plugging this result in Proposition \ref{prop:post_bounds_ante_bounds_worst}, we conclude that Ranking achieves an average stability index and both ex-ante and average optimality ratio of $1-\frac{1}{e}$ over all assignment games.
	
	Regarding $\lambda^{\text{post}}$, related to the competitive ratio of deterministic algorithms, \citet{aggarwal2011online} proved that no algorithm can do better than $1/2$ and that \textit{Greedy}, the algorithm that matches the arriving vertex to the most profitable neighbor. We extend this bound to $\kappa^{\text{post}}$ and $\kappa^{\text{ante}}$ by complementing the Greedy algorithm with our pricing method \textit{Half}, define in Proposition \ref{prop:randomized_half_prices}.
	
	\begin{proposition}\label{prop:vertex_weighted_worst_is_one_half}
		Consider the vertex-weighted setting. It follows,
		\begin{align*}
			\min_{\text{Instance } \Gamma} \lambda^{\text{post}}(\text{Greedy},\Gamma) = \min_{\text{Instance } \Gamma} \kappa^{\gamma} (\text{Greedy}+\textit{Half}, \Gamma) = \frac{1}{2},
		\end{align*}
		with $\gamma \in\{\text{post}, \text{ante}\}$ and moreover, no algorithm can do better in any of both cases.
	\end{proposition}
	
	\begin{proof}[Proof Sketch]
		The $\text{Greedy}+\textit{Half}$ algorithm (see Algorithm \ref{alg:greedy_half} in Appendix \ref{sec:missing proofs}) sets prices so for any pair of agents $(i,j) \in B\times S$, each of them obtains at least half of utility they generate together, obtaining that the algorithm produces allocations in the $\kappa$-approximated core with $\kappa\geq\nicefrac{1}{2}$. The tightness then is proved by exhibiting an instance where the algorithm achieves a value $\kappa$ exactly equal to $\nicefrac{1}{2}$.  
		The full proof is included in Appendix \ref{sec:missing proofs}.
	\end{proof}
	
	\noindent\Cref{tab:guarantees_vertex_weighted} summarizes our optimality and stability guarantees for the vertex-weighted setting.
	
	
	\begin{table}[H]
		\caption{Stability and optimality guarantees in vertex-weighted online assignment games. Remark all given values are tight, i.e., no algorithm can do better and for each of them, at least one algorithm achieves it. The results in bold correspond to our contributions, while the others are reformulations of literature results.}
		\label{tab:guarantees_vertex_weighted}
		\centering
		\begin{tabular}{c|c|c|c}
			& ex-post & ex-ante & avg\\
			\hline
			$\kappa$   & $\pmb{\nicefrac{1}{2}}$ & $\pmb{\nicefrac{1}{2}}$   & $1-\nicefrac{1}{e}$ \\
			$\normSI$  & $\pmb{\nicefrac{1}{2}}$ & ? &  $\boldsymbol{1-\nicefrac{1}{e}}$\\
			$\lambda$  & $\nicefrac{1}{2}$ &  $1-\nicefrac{1}{e}$     &  $1-\nicefrac{1}{e}$
		\end{tabular} 
	\end{table}
	Regarding the missing case, preliminary simulations suggest its value to be $1-\nicefrac{1}{e}$.
	\smallskip
	
	\noindent\textbf{Edge-weighted with free disposal}. The free disposal assumption was introduced by \citet{feldman2009online} due to the poor performances of online algorithms in the general edge-weighted model (Theorem 5.13 \cite{echenique_online_2023}). 
	In our terminology, their result states that for any randomized matching algorithm $\ALG$, it holds
	\begin{align*}
		\min_{\text{Instance } \Gamma}\lambda^{\text{avg}}(\ALG,\Gamma) = 0.
	\end{align*}
	In particular, from Proposition \ref{prop:post_bounds_ante_bounds_worst}, it follows that no randomized algorithm can achieve a constant factor for any of our guarantees.
	
	Under the free disposal assumption, in exchange, several works \cite{fisher2009analysis,lehmann2001combinatorial,blanc2022multiway} have proved the following guarantees, which we reformulate in our notation:
	
	\begin{proposition}\label{prop:lambda_guarantees_edge_weighted_free_disposal}
		Consider the edge-weighted with free disposal setting. For any randomized algorithm $\ALG$, it holds,
		\begin{align*}
			&\min_{\text{Instance } \Gamma} \lambda^{\text{post}}(\ALG,\Gamma) \leq \frac{1}{2},
		\end{align*}
		while, for any instance $\Gamma$, $\lambda^{\text{post}}(\text{Greedy},\Gamma) \geq \frac{1}{2}$. In addition, there exists a randomized algorithm $\ALG$ achieving,
		\begin{align*}
			\min_{\text{Instance } \Gamma} \lambda^{\text{ante}}(\ALG,\Gamma) = \min_{\text{Instance } \Gamma} \lambda^{\text{avg}}(\ALG,\Gamma) = 0.536.
		\end{align*}
	\end{proposition}
	
	Regarding the stability of these solutions, as for the edge arrival model, we obtain the following result.
	
	\begin{proposition}\label{prop:kappa_guarantees_edge_weighted_free_disposal}
		Consider the edge-weighted with free disposal setting. For any randomized algorithm $\ALG$, it holds,
		\begin{align*}
			&\min_{\text{Instance } \Gamma} \kappa^{\text{post}}(\ALG,\Gamma) = \min_{\text{Instance } \Gamma} \kappa^{\text{ante}}(\ALG,\Gamma) = 0.
		\end{align*}
	\end{proposition}
	
	\begin{proof}[Proof Sketch]
		The proposition is proved by constructing two vertex-arrival instances, each of them with two buyers and two sellers. Similarly to the examples showed in Proposition \ref{prop:tightness_half_prices}, depending on the decision of the algorithm on the first arriving buyer, we choose the instance where the algorithm makes a mistake. In particular, it at least one of the two, the algorithm achieves an allocation in the $0$-approximated core for $\kappa^\text{post}$. Regarding $\kappa^\text{ante}$, we construct an instance by considering several copies of the one created for $\kappa^\text{post}$. The full proof is included in Appendix \ref{sec:missing proofs}.
	\end{proof}
	
	The previous results, combined with Proposition \ref{prop:randomized_half_prices}, are summarized in Table~\ref{tab:guarantees_edge_weighted_with_free_disposal}.

	\begin{table}[H]
		\centering
		\caption{Stability and optimality guarantees in edge-weighted with free disposal online assignment games. Values with an inequality are lower bounds whose tightness remains open. Results in bold correspond to our contributions, while the others are reformulations of literature results.}
		\label{tab:guarantees_edge_weighted_with_free_disposal}
		\begin{tabular}{c|c|c|c}
			& ex-post & ex-ante & avg\\
			\hline
			$\kappa$   & $\boldsymbol{0}$  & $\boldsymbol{0}$ & ? \\
			$\normSI$   &  $\boldsymbol{\geq\nicefrac{1}{4}}$ & $\boldsymbol{\geq 0.268}$  &  $\boldsymbol{\geq 0.268}$\\
			$\lambda$   &  $\nicefrac{1}{2}$ &  $\geq 0.536$ &  $\geq 0.536$
		\end{tabular}
	\end{table}
	
	The value for $\kappa^{\text{avg}}$ remains unknown, although we conjecture the value is zero. Interestingly, when allowing side payments, the results of \citet{fahrbach2022edge} can be used to show that $\kappa^{\text{avg}} \bi \frac{1}{2}$.
	
	
	\section{Conclusions}\label{sec:conclusions}
	
	In this article, we initiated the study of stability in sub-optimal matchings and applied it to online assignment games, where either buyers or edges between buyers and sellers arrive sequentially. Our results show that stability naturally leads to optimality in the design of randomized algorithms for online matching, highlighting the study of stability in sub-optimal matchings as a promising and foundational research direction. 
	
	As a direction for future work, a formal study of the dynamics governing the evolution of the agents’ bargaining power (their ability to influence the split of the generated utility) and threat levels (the best utility an agent can obtain outside of their assigned match) in online settings, paralleling the static analysis of \citet{rochford_symmetrically_1984}, could provide valuable insights into the design of more stable allocation algorithms.

	\section*{Acknowledgments}
	
	\noindent This work was supported by the ANR LabEx CIMI (grant ANR-11-LABX-0040) within the French State Programme “Investissements d’Avenir.”
	\medskip
	
	\noindent Funded by the European Union. Views and opinions expressed are however those of the author(s) only and do not necessarily reflect those of the European Union or the European Research Council Executive Agency. Neither the European Union nor the granting authority can be held responsible for them. This work is supported by ERC grant 101166894 “Advancing Digital Democratic Innovation” (ADDI).
	\vspace{-0.5cm}
	\begin{figure}[H]
		\centering
		\includegraphics[scale = 0.2]{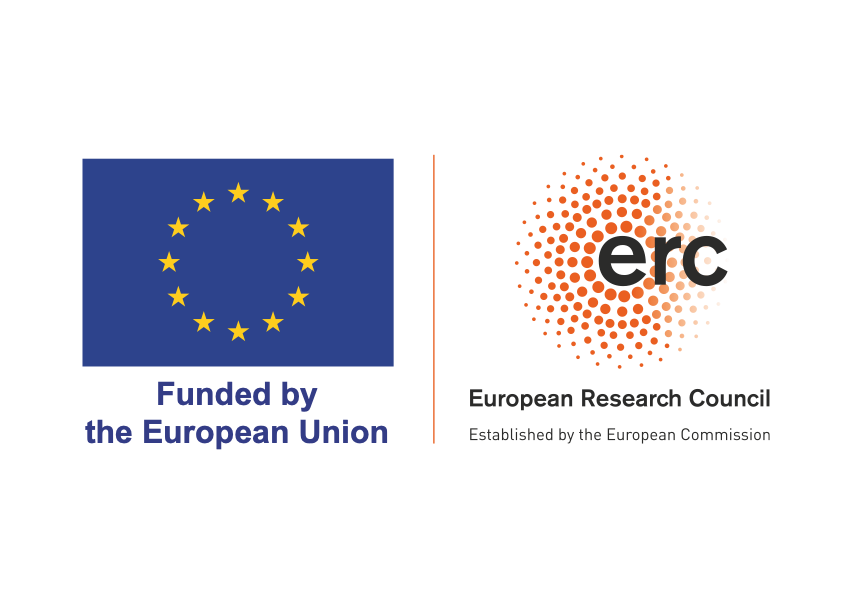}
	\end{figure}
	
	
	
	\bibliographystyle{ACM-Reference-Format} 
	\bibliography{biblio-report}

@article{shapley1971assignment,
  title={The assignment game I: The core},
  author={Shapley, Lloyd S and Shubik, Martin},
  journal={International Journal of game theory},
  volume={1},
  number={1},
  year={1971},
}

@article{jagadeesan2021learning,
  title={Learning equilibria in matching markets from bandit feedback},
  author={Jagadeesan, Meena and Wei, Alexander and Wang, Yixin and Jordan, Michael and Steinhardt, Jacob},
  journal={Advances in Neural Information Processing Systems (NeurIPS)},
  volume={34},
  year={2021}
}

@book{immorlica2023online,
  title={Online and Matching-Based Market Design},
  author={Immorlica, Nicole and Echenique, Federico and Vazirani, Vijay V},
  year={2023},
}

@inproceedings{freeman2020best,
  title={Best of both worlds: Ex-ante and ex-post fairness in resource allocation},
  author={Freeman, Rupert and Shah, Nisarg and Vaish, Rohit},
  booktitle={Proceedings of the 21st ACM Conference on Economics and Computation (EC)},
  year={2020}
}

@article{liu2014stable,
  title={Stable matching with incomplete information},
  author={Liu, Qingmin and Mailath, George J and Postlewaite, Andrew and Samuelson, Larry},
  journal={Econometrica: Journal of the Econometric Society},
  volume={82},
  number={2},
  year={2014},
}

@article{vazirani2022general,
  title={The general graph matching game: Approximate core},
  author={Vazirani, Vijay V},
  journal={Games and Economic Behavior},
  volume={132},
  year={2022},
}

@article{fahrbach2022edge,
  title={Edge-weighted online bipartite matching},
  author={Fahrbach, Matthew and Huang, Zhiyi and Tao, Runzhou and Zadimoghaddam, Morteza},
  journal={Journal of the ACM},
  volume={69},
  number={6},
  year={2022},
}

@article{garrido2025stable,
  title={Stable matching games},
  author={Garrido-Lucero, Felipe and Laraki, Rida},
  journal={Social Choice and Welfare},
  year={2025},
}

@inproceedings{cen2022regret,
  title={Regret, stability \& fairness in matching markets with bandit learners},
  author={Cen, Sarah H and Shah, Devavrat},
  booktitle={Proceedings of the International Conference on Artificial Intelligence and Statistics (AISTATS)},
  year={2022}
}

@String{Computing = "Computing" }

@String{Computer = "{IEEE} Computer" }

@article{demange_multi_1986,
	title = {Multi-item auctions},
	volume = {94},
	number = {4},
	journal = {Journal of Political Economy},
	author = {Demange, Gabrielle and Gale, David and Sotomayor, Marilda},
	year = {1986},
	}

@article{kelso_jr_job_1982,
	title = {Job matching, coalition formation, and gross substitutes},
	journal = {Econometrica: Journal of the Econometric Society},
	author = {Kelso Jr, Alexander S and Crawford, Vincent P},
	year = {1982},
	}

@article{demange_strategy_1985,
	title = {The strategy structure of two-sided matching markets},
	journal = {Econometrica: Journal of the Econometric Society},
	author = {Demange, Gabrielle and Gale, David},
	year = {1985},
	}

@article{crawford_job_1981,
	title = {Job matching with heterogeneous firms and workers},
	journal = {Econometrica: Journal of the Econometric Society},
	author = {Crawford, Vincent P and Knoer, Elsie Marie},
	year = {1981},
	}

@article{shioura_partnership_2017,
	title = {On the {Partnership} formation problem},
	volume = {2},
	journal = {Journal of Mechanism and Institution Design},
	author = {Shioura, Akiyoshi and {others}},
	year = {2017},
	}

@article{eriksson_stable_2001,
	title = {Stable outcomes of the roommate game with transferable utility},
	volume = {29},
	number = {4},
	journal = {International Journal of Game Theory},
	author = {Eriksson, Kimmo and Karlander, Johan},
	year = {2001},
	}

@article{echenique_core_2002,
  title={Core many-to-one matchings by fixed-point methods},
  author={Echenique, Federico and Oviedo, Jorge},
  journal={Journal of Economic Theory},
  volume={115},
  number={2},
  year={2004},
}

@article{rochford_symmetrically_1984,
	title = {Symmetrically pairwise-bargained allocations in an assignment market},
	volume = {34},
	number = {2},
	journal = {Journal of Economic Theory},
	author = {Rochford, Sharon C},
	year = {1984},
	}

@book{echenique_online_2023,
	title = {Online and {Matching}-{Based} {Market} {Design}},
	author = {Echenique, Federico and Immorlica, Nicole and Vazirani, Vijay},
	year = {2023},
}

@article{shi2023optimal,
  title={Optimal matchmaking strategy in two-sided marketplaces},
  author={Shi, Peng},
  journal={Management Science},
  volume={69},
  number={3},
    year={2023},
  }

@article{qiu2016approximate,
  title={Approximate core allocations and integrality gap for the bin packing game},
  author={Qiu, Xian and Kern, Walter},
  journal={Theoretical Computer Science},
  volume={627},
    year={2016},
  }

@article{faigle1998approximate,
  title={Approximate core allocation for binpacking games},
  author={Faigle, Ulrich and Kern, Walter},
  journal={SIAM Journal on Discrete Mathematics},
  volume={11},
  number={3},
    year={1998},
  }

@inproceedings{karp1990optimal,
  title={An optimal algorithm for on-line bipartite matching},
  author={Karp, Richard M and Vazirani, Umesh V and Vazirani, Vijay V},
  booktitle={Proceedings of the 22nd annual ACM symposium on Theory of computing (STOC)},
    year={1990}
}

@inproceedings{aggarwal2011online,
  title={Online vertex-weighted bipartite matching and single-bid budgeted allocations},
  author={Aggarwal, Gagan and Goel, Gagan and Karande, Chinmay and Mehta, Aranyak},
  booktitle={Proceedings of the 22nd Annual ACM-SIAM Symposium on Discrete Algorithms (SODA)},
    year={2011}
}

@inproceedings{feldman2009online,
  title={Online ad assignment with free disposal},
  author={Feldman, Jon and Korula, Nitish and Mirrokni, Vahab and Muthukrishnan, Shanmugavelayutham and P{\'a}l, Martin},
  booktitle={Proceedings of the International workshop on internet and network economics (WINE)},
    year={2009}
}

@incollection{fisher2009analysis,
  title={An analysis of approximations for maximizing submodular set functions—II},
  author={Fisher, Marshall L and Nemhauser, George L and Wolsey, Laurence A},
  booktitle={Polyhedral Combinatorics: Dedicated to the memory of DR Fulkerson},
    year={2009},
}

@inproceedings{lehmann2001combinatorial,
  title={Combinatorial auctions with decreasing marginal utilities},
  author={Lehmann, Benny and Lehmann, Daniel and Nisan, Noam},
  booktitle={Proceedings of the 3rd ACM conference on Electronic Commerce (EC)},
    year={2001}
}

@inproceedings{blanc2022multiway,
  title={Multiway online correlated selection},
  author={Blanc, Guy and Charikar, Moses},
  booktitle={Proceedings of the 62nd Annual Symposium on Foundations of Computer Science (FOCS)},
    year={2022}
}

@article{bullinger2025stability,
  title={Stability in online coalition formation},
  author={Bullinger, Martin and Romen, Ren{\'e}},
  journal={Journal of Artificial Intelligence Research},
  volume={82},
    year={2025}
}

@article{liu2023stability,
  title={Stability in repeated matching markets},
  author={Liu, Ce},
  journal={Theoretical Economics},
  volume={18},
  number={4},
    year={2023},
}

@article{doval2022dynamically,
  title={Dynamically stable matching},
  author={Doval, Laura},
  journal={Theoretical Economics},
  volume={17},
  number={2},
    year={2022},
}

@inproceedings{das2005two,
  title={Two-Sided Bandits and the Dating Market.},
  author={Das, Sanmay and Kamenica, Emir},
  booktitle={Proceedings of the International Joint Conferences on Artificial Intelligence (IJCAI)},
  volume={5},
    year={2005}
}

@inproceedings{liu2020competing,
  title={Competing bandits in matching markets},
  author={Liu, Lydia T and Mania, Horia and Jordan, Michael},
  booktitle={Proceedings of the International Conference on Artificial Intelligence and Statistics (AISTATS)},
    year={2020}
}

@inproceedings{basu2021beyond,
  title={Beyond $log^{2}(T)$ regret for decentralized bandits in matching markets},
  author={Basu, Soumya and Sankararaman, Karthik Abinav and Sankararaman, Abishek},
  booktitle={Proceedings of the International Conference on Machine Learning (ICML)},
    year={2021}
}

@article{bikhchandani2017stability,
  title={Stability with one-sided incomplete information},
  author={Bikhchandani, Sushil},
  journal={Journal of Economic Theory},
  volume={168},
    year={2017},
}

@article{min2022learn,
  title={Learn to match with no regret: Reinforcement learning in markov matching markets},
  author={Min, Yifei and Wang, Tianhao and Xu, Ruitu and Wang, Zhaoran and Jordan, Michael and Yang, Zhuoran},
  journal={Advances in Neural Information Processing Systems (NeurIPS)},
  volume={35},
    year={2022}
}

@article{mehta2013online,
  title={Online Matching and Ad Allocation},
  author={Mehta, Aranyak},
  journal={Foundations and Trends in Theoretical Computer Science},
  volume={8},
  number={4},
    year={2013},
}

@article{huang2024online,
  title={Online matching: A brief survey},
  author={Huang, Zhiyi and Tang, Zhihao Gavin and Wajc, David},
  journal={ACM SIGecom Exchanges},
  volume={22},
  number={1},
    year={2024},
}

@InProceedings{shin_et_al,
  author =	{Shin, Yongho and An, Hyung-Chan},
  title =	{{Making Three out of Two: Three-Way Online Correlated Selection}},
  booktitle =	{32nd International Symposium on Algorithms and Computation (ISAAC 2021)},
  pages =	{49:1--49:17},
  series =	{Leibniz International Proceedings in Informatics (LIPIcs)},
  ISBN =	{978-3-95977-214-3},
  ISSN =	{1868-8969},
  year =	{2021},
  volume =	{212},
  editor =	{Ahn, Hee-Kap and Sadakane, Kunihiko},
  publisher =	{Schloss Dagstuhl -- Leibniz-Zentrum f{\"u}r Informatik},
  address =	{Dagstuhl, Germany},
  URL =		{https://drops.dagstuhl.de/entities/document/10.4230/LIPIcs.ISAAC.2021.49},
  URN =		{urn:nbn:de:0030-drops-154822},
  doi =		{10.4230/LIPIcs.ISAAC.2021.49}
}

@inproceedings{gao2022improved,
  title={Improved online correlated selection},
  author={Gao, Ruiquan and He, Zhongtian and Huang, Zhiyi and Nie, Zipei and Yuan, Bijun and Zhong, Yan},
  booktitle={Proceedings of the 2021 IEEE 62nd Annual Symposium on Foundations of Computer Science (FOCS)},
  year={2022}
}

@article{coles1998marketplaces,
  title={Marketplaces and matching},
  author={Coles, Melvyn G and Smith, Eric},
  journal={International Economic Review},
  pages={239--254},
  year={1998},
  publisher={JSTOR}
}
	

	\appendix
	
	\section{Missing proofs}
	
	\label{sec:missing proofs}
	
	\newcommand{\smallexample}[3]{\raisebox{-0.5\height}{\begin{tikzpicture}
				\node[state, scale=0.7] (j1) {$\ja_{#1}$};
				\node[state, scale=0.7, below=0.6cm of j1] (j2) {$\jb_{#1}$};
				\node[state, scale=0.7, below right = 0.1cm and 1cm of j1] (i1) {$\ia_{#1}$};
				\node[state, scale=0.7, right = 1cm of i1] (i2) {$\ib_{#1}$};
				
				\draw (j1) edge[-, above] node{1} (i1);
				\draw (j2) edge[-, below] node{1} (i1);
				\ifnum#2=1 \draw (j1) edge[-, above, bend left] node{$#3$} (i2);
				\draw (j2) edge[-, below, bend right, white] node{$#3$} (i2);
				\else \draw (j2) edge[-, below, bend right] node{$#3$} (i2);
				\draw (j1) edge[-, above, bend left, white] node{$#3$} (i2);
				\fi
	\end{tikzpicture}}}

	\begin{lemma}
		\label{lem:IR_price_opt}
		Let $(\allocation)$ be an allocation. If, for every $(i,j)\in\mu$, $a_{i,j} \geq 0$, then there exists $\mathbf{p}'$ such that $(\mu, \mathbf{p}')$ is individually rational and $\SI(\allocation) \geq \SI(\mu, \mathbf{p}')$.
	\end{lemma}
	
	\begin{proof}
		We define, for $j\in S$, $$p'_j := \left\{\begin{array}{ll}
			p_j & \text{if } p_j\in[c_j, h_{\mu_j, j}] \\
			h_{\mu_j, j} & \text{if } p_j > h_{\mu_j, j}\\
			c_j & \text{if } p_j< c_j
		\end{array}\right.$$ with $h_{\mu_j, j} = c_j$ if $\mu_j=j$. Note that, for $j \in S$, as $a_{\mu_j, j} \geq 0$, $c_j \leq h_{\mu_j,j}$ and thus $\mathbf{p}'$ is well defined. Let us split $B \cup S$ in $A^+$, $A^-$ and $A^=$ as the agents whose utility respectively increases, decreases and stagnates when passing from $\mathbf p$ to $\mathbf p'$. Note that the utility of agents in $A^+$ is $0$, while the utility of agents in $A^-$ is the utility generated by their match. Finally, note that $\mu(A^+) = A^-$, as if $\mathbf p'$ increase someone utility, it decreases the utility of her match.
		
		Let $B', S' \subset B, S$ be a coalition and $\mu' \in \matching{B'}{S'}$. We need to show that $$\SW(\mu') - \SW_{|B',S'}(\mu, \mathbf p') \leq I(\allocation).$$
		To do so, let us consider the coalition $B'', S''$ defined as $B'' = B' \cup (A^+\cap B)$ and $S'' = S' \cup(A^+\cap S)$ with the matching $\mu'$ (where the new agents are unmatched). We then just need to prove that $$\SW(\mu') - \SW_{|B',S'}(\mu, \mathbf p') \leq \SW(\mu') - \SW_{|B'', S''}(\allocation).$$
		Denote
		$$ D :=\ \SW(\mu') - \SW_{|B'',S''}(\allocation) - \left(\SW(\mu') - \SW_{|B',S'}(\mu, \mathbf p')\right).$$
		It follows,
		\begin{align*}
			D =\ & \SW_{|B',S'}(\mu, \mathbf p') - \SW_{|B'',S''}(\allocation) \\ =\ & \sum\limits_{i \in B'} u_i(\mu, \mathbf p') \ + \sum\limits_{j\in S'} v_j(\mu, \mathbf p')\  - \sum\limits_{i \in B''} u_i(\mu, \mathbf p) \ - \sum\limits_{j\in S''} v_j(\mu, \mathbf p),
		\end{align*}
		we need to prove that $D\geq 0$. As the utility of agent in $A^=$ is the same at $(\mu, \mathbf p)$ and $(\mu, \mathbf p')$, we can remove them from the sum,
		\begin{align*}
			D = & \sum\limits_{i\in B'\setminus A^=} u_i(\mu, \mathbf p') \ + \sum\limits_{j\in S'\setminus A^=} v_j(\mu, \mathbf p') - \sum\limits_{i\in B''\setminus A^=} u_i(\mu, \mathbf p) \ - \sum\limits_{j\in S''\setminus A^=} v_j(\mu, \mathbf p).
		\end{align*}
		Then, we can split the sums on $A^+$ and $A^-$, as $(A^+, A^-, A^=)$ is a partition of the set of all agents. Therefore, we have
		\begin{align*}
			\sum\limits_{i\in B'\setminus A^=} u_i(\mu, \mathbf p') &= \sum\limits_{i\in B'\cap A^+} u_i(\mu, \mathbf p') \ + \sum\limits_{i\in B'\cap A^-} u_i(\mu, \mathbf p')\\
			&= \sum\limits_{i\in B'\cap A^+} 0 \ + \sum\limits_{i\in B'\cap A^-} a_{i,\mu_i}
			= \sum\limits_{i\in B'\cap A^-} a_{i,\mu_i}.
		\end{align*}
		Similarly,
		\begin{align*}
			\sum\limits_{j\in S'\setminus A^=} v_j(\mu, \mathbf p') &= \sum\limits_{j\in S'\cap A^+} v_j(\mu, \mathbf p') \ + \sum\limits_{j\in S'\cap A^-} v_j(\mu, \mathbf p')\\
			&= \sum\limits_{j\in S'\cap A^+} 0 \ + \sum\limits_{j\in S'\cap A^-} a_{\mu_j,j}
			= \sum\limits_{j\in S'\cap A^-} a_{\mu_j,j}.    
		\end{align*}
		Regarding $B''\setminus A^{=}$,
		\begin{align*}
			\sum\limits_{i\in B''\setminus A^=} u_i(\mu, \mathbf p') &= \sum\limits_{i\in B''\cap A^+} u_i(\mu, \mathbf p') \ + \sum\limits_{i\in B''\cap A^-} u_i(\mu, \mathbf p') \\
			&= \sum\limits_{i\in B\cap A^+} u_i(\mu, \mathbf p') \ + \sum\limits_{i\in B'\cap A^-} u_i(\mu, \mathbf p') \\
			&=\sum\limits_{j\in S\cap A^-} u_{\mu_j}(\mu, \mathbf p') \ + \sum\limits_{i\in B'\cap A^-} u_i(\mu, \mathbf p').
		\end{align*}
		Finally,
		\begin{align*}
			\sum\limits_{j\in S''\setminus A^=} v_j(\mu, \mathbf p') &= \sum\limits_{j\in S''\cap A^+} v_j(\mu, \mathbf p') \ + \sum\limits_{j\in S''\cap A^-} v_j(\mu, \mathbf p') \\
			&= \sum\limits_{j\in S \cap A^+} v_j(\mu, \mathbf p') \ + \sum\limits_{j\in S'\cap A^-} v_j(\mu, \mathbf p')\\
			&= \sum\limits_{i\in B \cap A^-} v_{\mu_i}(\mu, \mathbf p') \ + \sum\limits_{j\in S'\cap A^-} v_j(\mu, \mathbf p').
		\end{align*}
		The second equalities for the first two sums come from the fact that the utility of agents in $A^+$ is $0$, while the utility of agents in $A^-$ is the utility generated by their match. The second equalities for the two last sums come from the definition of $S''$ and $B''$ that allows to reindex the sum. The third equalities for the two last sums come from the fact that $\mu(A^+) = A^-$ and $\mu$ is a bijection. Gathering the sum with the same range of summation, we obtain
		\begin{align*}
			D = & \sum\limits_{i\in B'\cap A^-} a_{i,\mu_i} - u_i(\mu, \mathbf p) - v_{\mu_i}(\mu, \mathbf p) + \sum\limits_{j\in S'\cap A^-} a_{\mu_j, j} - u_{\mu_j}(\mu, \mathbf p) - v_j(\mu, \mathbf p)\\
			& - \sum\limits_{i \in (B \setminus B')\cap A^-} v_{\mu_i}(\mu, \mathbf p) \ - \sum\limits_{j \in (S \setminus S')\cap A^-} u_{\mu_j}(\mu, \mathbf p)\\
			= & - \sum\limits_{i \in (B \setminus B')\cap A^-} v_{\mu_i}(\mu, \mathbf p) \ - \sum\limits_{j \in (S \setminus S')\cap A^-} u_{\mu_j}(\mu, \mathbf p).
		\end{align*}
		Finally, we conclude that $D \geq 0$ as the match of agent in $A^-$ are in $A^+$, and thus were not individually rational at $(\allocation)$ and had thus negative utilities.
	\end{proof}
	
	\begin{lemma}
		Let $(\allocation)$ be an allocation. Then, the dual of
		\begin{align*}
			&\hspace{-1.5cm}\min_{(\boldsymbol{p},\boldsymbol{\tau},\boldsymbol{\eta})\in\mathbb{R}_+^{S}\times\mathbb{R}_+^{B}\times\mathbb{R}_+^{S}} \sum_{i\in B} \tau_i + \sum_{j\in S} \eta_j &\\
			\hspace{1cm}\text{s.t.}\ &h_{i,\mu_i} - p_{\mu_i} + \tau_i \geq 0, &\forall i \in B,\\
			&p_j - c_j + \eta_j \geq 0, &\forall j \in S,\\
			&h_{i,\mu_i} - p_{\mu_i} + \tau_i + p_j - c_j + \eta_j \geq a_{i,j}, &\forall (i,j) \in B \times S.
		\end{align*}
		can be written as
		\begin{align*}
			\max &\sum_{i\in B}\sum_{j\in S} \gamma_{i,j}(h_{i,j} - h_{i,\mu_i}) + \sum_{j\in S}\beta_j c_j  -\sum_{i\in B}\alpha_i h_{i,\mu_i} \\
			\text{s.t.}\ & \alpha_i + \sum_{j\in S}\gamma_{i,j} = 1 , \forall i \in B,\\
			&\beta_j + \sum_{i\in B}\gamma_{i,j} \leq 1, \forall j \in S,\\ &\boldsymbol{\alpha},\boldsymbol{\beta},\boldsymbol{\gamma} \geq 0.
		\end{align*}
	\end{lemma}
	
	\begin{proof}
		Set $\theta_i := \tau_i - p_{\mu_i} \in \mathbb{R}$ and $w_j := p_j + \eta_j\in\mathbb{R}_+$, and consider the Lagrangian,
		\begin{align*}
			\mathcal{L}(\boldsymbol{p}, \boldsymbol{\theta},\boldsymbol{w},\boldsymbol{\alpha},\boldsymbol{\beta},\boldsymbol{\gamma}) &:= \sum_{i\in B} (\theta_i + p_{\mu_i}) + \sum_{j\in S} (w_j - p_j) - \sum_{j\in S}\beta_j(w_j - c_j) \\
			&-\sum_{i\in B}\alpha_i (h_{i,\mu_i} + \theta_i)  - \sum_{i\in B}\sum_{j\in S}\gamma_{i,j}(h_{i,\mu_i} + \theta_i + w_j - c_j - a_{i,j}).
		\end{align*}
		It follows that 
		\begin{align*}
			\max_{\boldsymbol{\alpha},\boldsymbol{\beta},\boldsymbol{\gamma}}\min_{\boldsymbol{p}, \boldsymbol{\theta},\boldsymbol{w}}  \mathcal{L}(\boldsymbol{p}, \boldsymbol{\theta},\boldsymbol{w},\boldsymbol{\alpha},\boldsymbol{\beta},\boldsymbol{\gamma}) \leq 
			\min_{\boldsymbol{p}, \boldsymbol{\theta},\boldsymbol{w}} \max_{\boldsymbol{\alpha},\boldsymbol{\beta},\boldsymbol{\gamma}} \mathcal{L}(\boldsymbol{p}, \boldsymbol{\theta},\boldsymbol{w},\boldsymbol{\alpha},\boldsymbol{\beta},\boldsymbol{\gamma}),
		\end{align*}
		and for both problems to be feasible, we impose $\boldsymbol{\alpha},\boldsymbol{\beta},\boldsymbol{\gamma} \geq 0$. Considering that $a_{i,j} + c_j = h_{i,j}$ and rearranging the Lagrangian, we obtain,
		\begin{align*}
			\mathcal{L}(\boldsymbol{p}, \boldsymbol{\theta},\boldsymbol{w},\boldsymbol{\alpha},\boldsymbol{\beta},\boldsymbol{\gamma}) &=  \sum_{i\in B}\sum_{j\in S} \gamma_{i,j}(h_{i,j} - h_{i,\mu_i}) + \sum_{j\in S}\beta_j c_j  -\sum_{i\in B}\alpha_i h_{i,\mu_i}\\
			&+\sum_{i\in B}\theta_i(1-\alpha_i - \sum_{j\in S}\gamma_{i,j})  + \sum_{i\in B} p_{\mu_i}
			+ \sum_{j\in S} w_j(1-\beta_j-\sum_{i\in B}\gamma_{i,j}) - \sum_{j\in S} p_j.
		\end{align*}
		Since $\sum_{i\in B} p_{\mu_i} = \sum_{j\in S} p_j$, we obtain the dual linear problem,
		\begin{align*}
			\max_{\boldsymbol{\alpha},\boldsymbol{\beta},\boldsymbol{\gamma} \geq 0} &\sum_{i\in B}\sum_{j\in S} \gamma_{i,j}(h_{i,j} - h_{i,\mu_i}) + \sum_{j\in S}\beta_j c_j  -\sum_{i\in B}\alpha_i h_{i,\mu_i} \\
			\text{s.t.}\ & \alpha_i + \sum_{j\in S}\gamma_{i,j} = 1 , \forall i \in B,\\
			&\beta_j + \sum_{i\in B}\gamma_{i,j} \leq 1, \forall j \in S. 
		\end{align*}
	\end{proof}
	
	\begin{proof}[Proof of Proposition \ref{prop:post_bounds_ante_bounds_worst}]
		For any random variable with finite expectation, the minimum realization is a lower bound of the expectancy. This gives us the inequalities between the ex-ante measures and the ex-post ones. For the inequalities between the ex-ante and the average measures, we will prove them using the linearity of the expectation and the Jensen's inequality. For the sake of concision, we will solely denote $\mathbb{E}$ for $\mathbb{E}_{(\allocation)\sim ALG(\Gamma)}$.
		\begin{align*}
			\lambda^{\text{avg}}(\ALG, \Gamma) & = \frac{\sum_{i \in B} {u_i(\ALG,\Gamma)} \enspace + \enspace \sum_{j \in S} {v_j(\ALG,\Gamma)}}{\opt}\\
			&= \mathbb{E}\left[\frac{\sum_{i \in B} u_i(\allocation) \enspace + \enspace \sum_{j \in S} v_j(\allocation)}{\opt}\right]
			\\ & = \mathbb{E}\left[\lambda(\mu)\right] = \lambda^{\text{ante}}(\ALG,\Gamma).
		\end{align*}
		Regarding $\kappa$,
		\begin{align*}
			\kappa^{\text{avg}}(\ALG,\Gamma) & = \min_{(i,j)\in B\times S} \frac{{u_i(\ALG,\Gamma)} + {v_j(\ALG,\Gamma)}}{a_{i,j}}\\
			&= \min_{(i,j)\in B\times S} \mathbb{E}\left[\frac{u_i(\allocation) + v_j(\allocation)}{a_{i,j}}\right]
			\\ & \geq \mathbb{E}\left[\min_{(i,j)\in B\times S} \frac{u_i(\allocation) + v_j(\allocation)}{a_{i,j}}\right]
			\\&= \mathbb{E}\left[\kappa(\allocation)\right] = \kappa^{\text{ante}}(\ALG,\Gamma).
		\end{align*}
		Finally, for $\normSI^{\text{ante}} \leq \normSI^{\text{avg}}$, thanks to the linearity of $\mathbb E$, we just need to prove that $\esp{\SI(\allocation)} \geq \SI^{\text{avg}}$, with $\SI^{\text{avg}}$ the subset instability when replacing utilities with $(u(\ALG, \Gamma), v(\ALG, \Gamma))$. 
		\begin{align*}
			\SI^{\text{avg}} & = \max\limits_{\substack{(B',S')\subseteq (B, S) \\ \mu' \in \matching{B'}{S'}}} \left\{ \SW(\mu') - \left( \sum\limits_{i\in B'} u_i(\ALG, \Gamma) + \sum\limits_{j\in S'} v_j(\ALG, \Gamma) \right)\right\} \\
			& = \max\limits_{\substack{(B',S')\subseteq (B, S) \\ \mu' \in \matching{B'}{S'}}}\left\{ \mathbb{E}\left[\SW(\mu') - \left( \sum\limits_{i\in B'} u_i(\allocation)+ \sum\limits_{j\in S'}v_j(\allocation)\right)\right]\right\} \\
			& \leq \mathbb{E}\left[\max\limits_{\substack{(B',S')\subseteq (B, S) \\ \mu' \in \matching{B'}{S'}}}\left\{ \SW(\mu') - \left( \sum\limits_{i\in B'} u_i(\allocation)+ \sum\limits_{j\in S'}v_j(\allocation)\right)\right\}\right] \\
			& = \esp{\SI(\allocation)}.
		\end{align*}
		
		The preservation of \Cref{eq:three_metrics} for the ex-post and ex-ante measures comes from the monotonicity of the minimum and the expectation. For the average measures, the proof are similar to the ones leading to \Cref{eq:three_metrics} but with $(u(\ALG, \Gamma), v(\ALG, \Gamma))$ instead of $(u,v)$.
	\end{proof}
	
	\begin{proof}[Proof of Proposition \ref{prop:randomized_half_prices}]
		\newcommand{\probphalf}{\hat{\mathbf p}^{\text{half}}}
		Let us be more formal, adopting probabilistic notation. Let $\Omega$ be the universe. We define $(\probmu, \probphalf) : \Omega \to \matching{B}{S}\times \mathbb R^S$ as the result of $\ALG + \textit{Half}$ on $\Gamma$. Then, for every $\omega \in \Omega$, $(\probmu(\omega), \probphalf(\omega))$ is an allocation with prices set as in Proposition \ref{prop:half_prices_is_half_stable}. Thus, for every $\omega \in \Omega$, $\frac{1}{2}\cdot \lambda(\probmu (\omega)) \leq \normSI(\probmu(\omega), \probphalf(\omega))$. By linearity and monotonicity of the expectation, we obtain that 
		\begin{align*} 
			& \frac{1}{2}\cdot \lambda^{\text{ante}}(\ALG + \textit{Half}, \Gamma) = \frac{1}{2}\cdot \esp{\lambda(\probmu)} = \esp{\frac{1}{2}\cdot \lambda(\probmu )} \leq \esp{\normSI(\probmu, \probphalf)} = \normSI^{\text{ante}}(\ALG + \textit{Half}, \Gamma).
		\end{align*}
		Similarly, by linearity and monotonicity of the minimum, we obtain that $\frac{1}{2}\cdot \lambda^{\text{post}}(\probmu) \leq \normSI^{\text{post}}(\probmu, \probphalf)$, ex-post guarantees being the minimum over all $\omega\in\Omega$.
		
		Finally, we need to prove the inequality for average guarantees. We will then adapt the proof of Proposition \ref{prop:half_prices_is_half_stable}. Let $(B', S') \subset (B, S)$ and $\mu' \in \matching{B'}{S'}$.  As stated in the proof of Proposition \ref{prop:half_prices_is_half_stable}, for $\omega \in \Omega$, $$\SW(\mu') - \SW_{|B',S'}(\probmu(\omega), \probphalf(\omega)) \leq \opt - \frac{1}{2} \SW(\probmu(\omega)). $$
		By linearity and monotonicity of the expectation, we obtain
		$$ \SW(\mu') - \sum_{i \in B'} \esp{u_i(\probmu, \probphalf)} - \sum_{j\in S'} \esp{v_j(\probmu, \probphalf)} \leq \opt - \frac{1}{2}\cdot \esp{\SW(\probmu)}.$$
		Thus, taking the maximum over all coalition, we obtain
		\begin{align*}
			I^{\text{avg}} := & \max_{\substack{(B', S') \in (B, S) \\ \mu' \in \matching{B'}{S'}} } \SW(\mu') - \sum_{i \in B'} \esp{u_i(\probmu, \probphalf)} - \sum_{j\in S'} \esp{v_j(\probmu, \probphalf)} \leq \opt - \frac{1}{2}\cdot \esp{\SW(\probmu)},
		\end{align*}
		which gives \begin{align*}
			\normSI^{\text{avg}}(\ALG + \textit{Half}, \Gamma) &= 1 - \frac{ I^{\text{avg}}}{\opt} \geq \frac{1}{2} \cdot \frac{\esp{\SW(\probmu)}}{\opt}  = \frac{1}{2} \cdot \lambda^{\text{ante}}(\ALG + \textit{Half}, \Gamma)\\
			&= \frac{1}{2} \cdot \lambda^{\text{avg}}(\ALG + \textit{Half}, \Gamma).
		\end{align*}
	\end{proof}
	
	Let us formalize the algorithm ranking, assuming that the value of seller $j$ is $a_j$ and using the formulation from \cite{echenique_online_2023}.
	
	\begin{algorithm}[H]
		\caption{$\textit{Ranking}$}
		\label{alg:ranking}
		\begin{algorithmic}
			\FOR{$j \in S$}
			\STATE{Set price $p_j \gets a_je^{w_j-1}$ where $w_j \in[0,1]$ is sampled uniformly.}
			\ENDFOR
			\FOR{each buyer $i$ who arrives}
			\STATE{Let $N(i)$ be the unmatched neighbors of $i$.}
			\IF{$N(i) \neq \emptyset$}
			\STATE{Match $i$ to $j\in N(i)$ maximizing $a_j-p_j$.}
			\ENDIF
			\ENDFOR
		\end{algorithmic}
	\end{algorithm}
	
	\begin{proof}[Proof of Proposition \ref{prop:vertex_weighted_worst_is_one_half}]
		\newcommand{\algodemi}{\textit{Greedy}+\textit{Half}}
		Let us first formalize the algorithm denoted $\algodemi$. Any pair of agents that generate a non-zero utility are referred to as neighbors.
		
		\begin{algorithm}[H]
			\caption{$\algodemi$}
			\label{alg:greedy_half}
			\begin{algorithmic}
				\STATE{Start with an empty allocation}
				\FOR{each buyer $i$ who arrives}
				\IF{$i$ has no unmatched neighbors}
				\STATE{Leave $i$ unmatched}
				\ELSE
				\STATE{$j \gets \argmax\{ a_{i,j} \mid j\in B \text{ unmatched neighbor of } i \}$}
				\STATE{Match $i$ and $j$}
				\STATE{$p_j \gets c_j + \frac{a_{i,j}}{2}$}
				\ENDIF
				\ENDFOR
			\end{algorithmic}
		\end{algorithm}
		
		Let $(\allocation)$ be the result of $\algodemi$ on an instance $\Gamma$ of the online vertex weighted problem. In this context, let us note for $j \in S$, $a_j > 0$ such that for any $i\in B$, $a_{i,j}\in \{0, a_{j}\}$. First, notice that $(\allocation)$ is individually rational. Then, let $(i,j) \in B\times S$.
		
		If $a_{i,j} = 0$, (i.e. $i$ and $j$ are not neighbors), then $u_i(\allocation) + v_j(\allocation) \geq \frac{1}{2} \cdot 0$.
		
		Else, $a_{i,j} = a_j$. If $\mu_i = i$ and $\mu_j = j$, then $\algodemi$ would have match them. Thus at least one of them is matched. If $j$ is matched, then $v_j(\allocation) = \frac{a_j}{2}$ and thus $u_i(\allocation) + v_j(\allocation) \geq  v_j(\allocation) \geq \frac{a_j}{2} = \frac{a_{i,j}}{2}$. Otherwise, $i$ is matched and $a_{\mu_i} \geq a_{j}$, as $j$ is unmatched and $\algodemi$ did not match $i$ and $j$. Thus, $u_i(\allocation) + v_j(\allocation) \geq u_i(\allocation) = \frac{a_{\mu_i}}{2} \geq \frac{a_{j}}{2} =  \frac{a_{i,j}}{2}$.
		
		Thus, $\kappa^{\text{post}}(\algodemi, \Gamma) = \kappa(\allocation) \geq \frac{1}{2}$.
		
		Furthermore, as proved in \citet{aggarwal2011online}, no better ratio can be obtained for $\lambda^{\text{post}}$. This can be seen on the instances represented in \Cref{fig:counter_example_VW_post} where, whatever the way an algorithm matches $\ia$ with a positive probability, on one of the two instances, the outputted allocation will be $\frac{1}{2}$-optimal. Moreover, the same result holds if $\ia$ is never matched. 
		
		Finally, to obtain the upper bound for $\kappa^{\text{ante}}$, we take multiple independent instances of \Cref{fig:counter_example_VW_post} so that with high probability, on at least one of the sub-instances, $\kappa$ will be $\frac{1}{2}$. The formal construction of this instances is similar to the one of the proof of Proposition \ref{prop:kappa_guarantees_edge_weighted_free_disposal}. We conclude the proof by using Proposition \ref{prop:post_bounds_ante_bounds_worst}.
		\begin{figure}
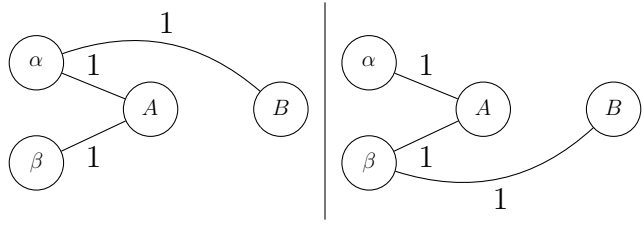

			\centering
			\begin{tabular}{c|c} \smallexample{}{1}{1} &  \smallexample{}{2}{1} \end{tabular}
			\caption{Instances on which no ex-post guarantees better than $\nicefrac{1}{2}$ can be achieved.}
			\label{fig:counter_example_VW_post}
		\end{figure}
	\end{proof}
	
	\begin{proof}[Proof of Proposition \ref{prop:kappa_guarantees_edge_weighted_free_disposal}]
		
		Let $\ALG$ be an algorithm for the edge-weighted with free disposal setting. Let $l\in\mathbb N^*$ and $W \in \mathbb{R}^*_+$. Let us consider the instance $\Gamma^{begin}_{W, l}$ as \begin{itemize}
			\item $B = \{\ia_1, \ia_2, \dots, \ia_l\}$
			\item $S = \{\ja_1, \jb_1, \ja_2, \jb_2, \dots, \ja_l, \jb_l\}$
			\item for $k\in\{1, \dots, l\}$, $a_{\ia_k, \ja_k} = a_{\ia_k, \jb_k} = 1$
			\item $a = 0$ elsewhere
			\item buyers arrives in the order $\ia_1, \ia_2, \dots, \ia_l$.
		\end{itemize}
		
		\begin{figure}[ht]
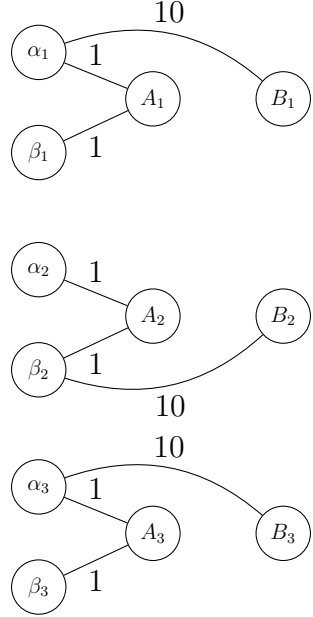

			\centering
			\begin{tabular}{c}
				\smallexample{1}{1}{10}\\
				\smallexample{2}{2}{10}
				\\ \smallexample{3}{1}{10}
			\end{tabular}\\
			Order of arrival : $\ia_1, \ia_2, \ia_3, \ib_1, \ib_2, \ib_3$.
			\caption{Example of hard instance for the edge-weighted problem with free disposal. In this case, $l = 3$ and $W = 10$.}
			\label{fig:counter_example_EW-FD}
		\end{figure}
		Let us consider $\hat{\mu}^{begin}$ the random variable representing the matching outputted by $\mathcal A$ on $\Gamma^{begin}_{W, l}$. Then, in each pair $(\ja_k, \jb_k)$, one seller is chosen more than the other. To make this choice a bad one, new buyers will arrive, with an high utility with the sellers that are more likely matched. An example of such instance is presented in Figure \ref{fig:counter_example_EW-FD}. Formally, and to avoid correlation problems, we will find a sequence $(x_1, \dots, x_l) \in \prod_{1\leq k\leq l} \{\ja_k, \jb_k\}$ not chosen enough.
		\begin{align*}
			1 &\geq \mathbb P\biggl( \bigsqcup\limits_{\substack{(x_1, \dots, x_l) \\ x_k\in\{\ia_k, \ib_k\}}} \Big( \forall k \in \{1, \dots, l\},\, \hat{\mu}^{begin}(\ia_k) = x_k  \Big) \biggr)\\
			&= \sum\limits_{\substack{(x_1, \dots, x_l) \\ x_k\in\{\ia_k, \ib_k\}}} \mathbb P\Big( \forall k \in \{1, \dots, l\},\,  \hat{\mu}^{begin}(\ia_k) = x_k  \Big).
		\end{align*}
		
		Thus, as it is a sum of $\frac{1}{2^l}$ terms that is smaller than $1$, 
		\begin{align*}\exists (x_1, \dots,& x_l) \in \prod\limits_{1\leq k \leq l} \{\ja_k, \jb_k\} \: :  \mathbb{P}\Big( \forall k \in \{1, \dots, l\},\,  \hat{\mu}^{begin}(\ia_k) = x_k \Big) \leq \frac{1}{2^l}. \end{align*}
		Then, for $k \in \{1, \dots, l\}$, let us denote $y_k := \ja_k$ if $x_k = \jb_k$ and $y_k:= \jb_k$ if $x_k = \ja_k$. Intuitively, $x$ are the sellers not enough matched, while $y$ are the sellers too matched, which are going to have better opportunities. Then, we extend $\Gamma^{begin}_{W, l}$ in $\Gamma^{full}_{W, l}$ such that $l$ new buyers arrive, (after the $l$ previous ones), $\ib_1, \dots, \ib_l$, with $a_{\ib_k, y_k} = W$ and $a_{\ib_k, z}= 0$ for $1 \leq k \leq l$ and $z \in S \backslash \{y_k\}$. For $k\in\{1, \dots, l\}$, if $\ia_k$ and $y_k$ are matched, then if we do not match $\ib_k$, then $(y_k, \ib_k)$ will be able to produce $W$ while having at most $1$ and if we match $y_k$ and $B_k$, $(x_k, \ja_k)$ will have $0$ while being able to produce $1$. Whatever the case is, (and whatever the prices), if for any $k\in\{1, \dots, l\}$, $\ia_k$ and $y_k$ are matched, $\kappa$ will be smaller than $\frac{1}{W}$. Then, $\ALG$ being online, it starts by producing $\hat{\mu}^{begin}$ on $\Gamma^{full}_{W, l}$, and thus,
		\begin{align*}
			\mathbb P\left(\kappa\left(\ALG\left(\Gamma^{full}_{W, l}\right)\right) \leq \frac{1}{W} \right) & \geq \mathbb P\left( \exists i \in\{1, \dots, l\} \: : \: \hat{\mu}^{begin}(\ia_i) = y_i \right)\\
			& \geq 1 - \mathbb P\left( \forall i \in\{1, \dots, l\} \: : \: \hat{\mu}^{begin}(\ia_i) = x_i \right)\\
			& \geq 1 - \frac{1}{2^l}
		\end{align*}
		with the first inequality coming from the discussion above, the second one being an inequality and not an equality because $\ia_i$ may have been left unmatched and the last being the definition of $x$.
		Finally, \begin{align*}
			\kappa^{\text{ante}} (\ALG, \Gamma^{full}_{W, l}) &= \mathbb E_{(\allocation)\sim \ALG(\Gamma^{full}_{W, l})}[\kappa(\allocation)] \\
			& \leq \frac{1}{W} \cdot \mathbb P\left(\kappa(\mathcal A(\Gamma^{full}_{W, l})) \leq \frac{1}{W} \right) + 1\cdot \mathbb P\left(\kappa(\ALG(\Gamma^{full}_{W, l})) > \frac{1}{W} \right) \\
			& = \frac{1}{W} + 1 -  \mathbb P\left(\kappa(\ALG(\Gamma^{full}_{W, l})) \leq \dfrac{1}{W} \right)\\
			&\leq \frac{1}{W} + \frac{1}{2^l}.
		\end{align*}
		This being true for every $l$ and $W$, making them both go to infinity prove the result.
	\end{proof}


	
\end{document}